\documentclass[reqno, 11pt]{article}

\usepackage{graphicx,tikz}%
\usepackage{multirow}%
\usepackage{amsmath,amssymb,amsfonts}%
\usepackage{amsthm}%
\usepackage{mathrsfs}%
\usepackage[title]{appendix}%
\usepackage{xcolor}%
\usepackage{textcomp}%
\usepackage{manyfoot}%
\usepackage{booktabs}%
\usepackage{algorithm}%
\usepackage{algorithmicx}%
\usepackage{algpseudocode}%
\usepackage{listings}%
\usepackage{chemformula}
\usepackage{siunitx}
\usepackage{authblk}
\usepackage{amssymb,latexsym,amsmath,amsfonts}
\usepackage{graphicx}
\usepackage{color,enumerate}
\usepackage{amsmath,amsfonts,amssymb}
\usepackage[mathscr]{eucal}
\usepackage{bbm,pgfkeys}
\usepackage{booktabs}
\usepackage{float,xcolor}
\usepackage{hyperref}
\usepackage{amsthm}
\usepackage[mathscr]{eucal}
\usepackage{paralist}
\usepackage{rotating}
\usepackage{multirow}
\usepackage[left=1.25in, right=1.25in,bottom=1.25in]{geometry}
\hypersetup{
	colorlinks=true,         
	linkcolor=red,           
	citecolor=blue,          
	urlcolor=blue,           
}

\usepackage{adjustbox}

\definecolor{medgreen}{rgb}{0, 0.75, 0}
\definecolor{darkgreen}{rgb}{0, 0.3, 0}

\newtheorem{theorem}{Theorem}

\newtheorem{proposition}{Proposition}

\theoremstyle{definition}

\theoremstyle{remark}

\newcommand{\RN}[1]{\uppercase\expandafter{\romannumeral#1}}

\DeclareMathOperator{\Tr}{Tr}

\title{Spatial modeling of forest-savanna bistability: Impacts of fire dynamics and timescale separation}

\author[1,2]{Kimberly Shen}
\affil[1]{Department of Physics, Princeton University, USA}
\affil[2]{Peabody Institute, Johns Hopkins University, USA}

\author[4,5]{Simon Levin}

\author[3,4,5]{Denis D. Patterson\thanks{Corresponding author: denis.d.patterson@durham.ac.uk}}
\affil[3]{Department of Mathematical Sciences, Durham University, UK}
\affil[4]{Department of Ecology and Evolutionary Biology, Princeton University, USA}
\affil[5]{High Meadows Environmental Institute, Princeton University, USA}

\date{\today}

\begin{document}


\maketitle

\begin{abstract}
Forest-savanna bistability -- the hypothesis that forests and savannas exist as alternative stable states in the tropics -- and its implications are key challenges for mathematical modelers and ecologists in the context of ongoing climate change. To generate new insights into this problem, we present a spatial Markov jump process model of savanna forest fires that integrates key ecological processes, including seed dispersal, fire spread, and non-linear vegetation flammability. In contrast to many models of forest-savanna bistability, we explicitly model both fire dynamics and vegetation regrowth in a mathematically tractable framework. This approach bridges the gap between slow-timescale vegetation models and highly resolved fire dynamics, shedding light on the influence of short-term and transient processes on vegetation cover. In our spatial stochastic model, bistability arises from periodic fires that maintain low forest cover, whereas dense forest areas inhibit fire spread and preserve high tree density. The deterministic mean-field approximation of the model similarly predicts bistability, but deviates quantitatively from the fully spatial model, especially in terms of its transient dynamics. These results also underscore the critical role of timescale separation between fire and vegetation processes in shaping ecosystem structure and resilience.
\end{abstract}

\section{Introduction}
Satellite data of tree cover in Sub-Saharan Africa combined with mean annual rainfall (MAR) records have shown that tree cover exhibits a distinctly bimodal distribution in regions with intermediate MAR (1000 to 2000 mm/yr). Savannas (regions with sparse open-canopy tree cover less than $\approx 40 \%$ per unit area) and forests (regions with high closed-canopy tree cover above $\approx 70 \%$ per unit area) are frequently observed in the tropics, while intermediate tree cover ($40$ to $70 \%$) is rarely observed~\cite{van2018fire}. The notion of forests and savannas existing as alternative stable states has attracted a significant amount of attention and concern since it suggests that perturbations caused by climate change, drought, or human activity can cause large-scale conversion of forests to savanna, which may be extremely difficult to reverse due to hysteresis effects~\cite{cochrane1999positive}. 

The preservation of forests, in turn, is a crucial concern in mitigating climate change, as forests have significantly higher carbon storage capacities than savannas. Forests also support local human populations by providing building materials. In addition to forests, the maintenance of savannas is also critical to biodiversity and conservation efforts due to the high species richness of ecosystems located at the savanna-forest ecotone~\cite{behling2007late,myers2000biodiversity,overbeck2007brazil}. Therefore, a model that accurately predicts the responses of forest-savanna mosaics to various environmental disturbances would usefully inform various climate- and biodiversity-related interventions, such as selecting and maintaining protected areas and regulations of human ecological disturbances. As such, a maximally valuable model must encompass short-term disturbances, such as fires, and long-term effects, such as forest encroachment or loss. The model ideally must also be spatially explicit since the fundamental driving processes of the system (e.g., seed dispersal and fire spread) have a highly spatially inhomogeneous structure.

A widely accepted explanation for the observed bimodal distribution of forests and savannas is that regular disturbance by fires can maintain open-canopy savannas. In contrast, denser closed-canopy forests sufficiently suppress the flammable grassy layer, which limits fire spread and lowers the probability of large fires, stabilizing the forest state. Mechanistically, fires typically spread readily within savannas fueled by highly flammable grasses but do not significantly penetrate forest boundaries~\cite{biddulph1998fuels,bond2008limits,hennenberg2008detection,hoffmann2009tree,kellman1997fire, swaine1992effects}. Thus, regions with sufficiently high grass cover may burn frequently enough to maintain savanna conditions despite climatological conditions supporting forest growth~\cite{brando2014abrupt,brando2012fire,cochrane2002fire}. The original Staver-Levin model~\cite{staver2011tree,staver2012integrating} and related models~\cite{accatino2010tree,tilman1994competition}, which incorporate fire-vegetation interactions in a spatially implicit mean-field approximation, have been widely used to explain and study forest-savanna bistability~\cite{staal2018resilience,van2014tipping,wuyts2017amazonian,yatat2018tribute}. We will use an extension of the Staver-Levin modeling framework as the basis for our study, but our conclusions apply more broadly across this class of spatially implicit forest-savanna models.

Many forest fire models, such as the Staver-Levin model, use mean-field approximations to keep them interpretable and analytically tractable and only indirectly incorporate spatial processes. For example,~\cite{hoyer2021impulsive,klimasara2023model,magnani2023fire} are recently proposed models that use a mean-field (i.e. spatially implicit) approach but add discontinuous, stochastic changes in vegetation proportions to implicitly account for fire spreading events. However, explicitly spatial models have some advantages over non-spatial models, which implicitly incorporate spatial effects~\cite{deangelis2017spatially}. For instance, spatial models have greater practical applicability and more accurately replicate the behavior of systems in which local mechanisms amplified by positive feedback, such as fire, can cause large-scale effects. Thus, to construct an accurate and practically applicable model, it is necessary to include sufficient detail to model both local conditions and the dynamics that they generate. Spatial structure can also profoundly affect the behavior of a system. For instance,~\cite{durrett2018heterogeneous} found that in a spatial model of grass-forest competition, both states can co-exist only when transition rates are spatially heterogeneous. 

Various other spatial models of fires in savanna-forest systems have already been proposed and studied, but have shortcomings. For example, some forest fire models are parametrized based on empirical data~\cite{aleman2018spatial,archibald2012evolution,lasslop2016multiple,rammer2019scalable}, but such models cannot transparently demonstrate the underlying dynamics of the system nor be applied to novel situations. Other studies based on field sampling of species diversity, fire frequency and/or soil quality have similar drawbacks~\cite{beckett2022pathways,sagang2024interactions}. Some mean-field models have been extended to incorporate spatial effects by adding diffusion terms to the nonspatial dynamics~\cite{goel2020dispersal,zelnik2024savanna}. However, these models are limited by an inability to model anisotropic forest seed dispersal and do not explicitly model fire dynamics. Explicit fire dynamics are essential for examining the behavior of grass-forest systems on shorter timescales and studying the effects of seasonality-dependent flammability on forest dynamics~\cite{scholes1997tree}. Nonetheless, these studies have already produced interesting results, such as the dependence of forest front travel speed on the curvature of the front and the need for a critical patch size to ensure growth, adding a hysteresis-like quality to the system's behavior~\cite{durrett2018heterogeneous, goel2020dispersal}. 

Most spatially explicit forest fire models employ a square grid and use nearest neighbor-dependent transition probabilities to model spatial processes such as fire spread and seed dispersal~\cite{abades2014fire,accatino2016trees, fair2020spatial,hochberg1994influences,li2019spatial,staal2018resilience,van2018fire,wuyts2023emergent}. These models have already led to notable findings. For example,~\cite{wuyts2023emergent} explicitly demonstrated that fire-vegetation feedback can maintain bistability in tree cover states,~\cite{fair2020spatial} found that forest-savanna mosaics occur under limited conditions, suggesting that such mosaics are vulnerable to loss under climate change, and~\cite{abades2014fire} found a second order phase transition in spanning cluster formation probability. However, these models have some intrinsic limitations. For instance, they cannot incorporate inhomogeneities in vegetation density and long-distance or anisotropic spreading of fire and seeds. Additionally, the ability of these models to explain observations on short time horizons may be limited by the relatively short timescale of fire dynamics compared to the timestep increment (typically 1 year).  

The paper is organized as follows: Section \ref{sec.model} introduces the spatial stochastic modeling framework that forms the basis for this study, including discussion on plausible parameter regimes and rigorous links to continuum models. Section \ref{sec.nonspatial_model} provides a bifurcation analysis of a nonspatial version of the model for later comparison with the spatial version and other similar nonspatial forest-savanna models in the extant literature. Section \ref{sec.spatial} studies the dynamics of the full spatial stochastic model and contrasts them with the dynamics of the nonspatial model from the previous section, emphasizing the key role of transient dynamics on short timescales. Finally, Section \ref{sec.conclusion} concludes with a synthesis of the main findings and highlights some future directions and open problems in the mathematical modeling of forest-savanna bistability.

\section{A Spatial Stochastic Forest Fire Model}\label{sec.model}
Ecologically, tropical forests differ markedly from forests of coniferous trees in several key ways. First, savanna trees are not readily killed by fire. Typically, they are only top-killed (i.e. only aerial biomass is burned) and can readily resprout~\cite{balch2011size,hoffmann2009tree,hoffmann2003comparative} or have thickened bark to prevent stem death~\cite{hoffmann2012ecological}. Second, fire does not propagate readily through tropical forests~\cite{archibald2009limits,hennenberg2006phytomass,pueyo2010testing} but does propagate rapidly through savanna grassland~\cite{wragg2018forbs}. Mechanistically, this is due to forest understory shade excluding flammable C4 grasses~\cite{hoffmann2012ecological} in addition to reduced wind speeds and increased moisture in the forest microclimate~\cite{hennenberg2008detection,hoffmann2012fuels}. Lastly, tropical forests have relatively open canopies that do not entirely shade out grass~\cite{hennenberg2006phytomass}. For brevity, we will henceforth refer to tropical forests as ``forest'' and savanna regions with lower tree density as ``grassland''. 

\subsection{Mathematical Description of the Model}
We use a spatially extended Markov jump process to model the dynamics of fires in forest-grass systems using the mathematical framework introduced in \cite{patterson2020probabilistic}. The set $\Omega \subset \mathbb{R}^2$ denotes the spatial domain and the random variables $\{r_i \in \Omega\,:\, i = 1,\dots, N \}$ are sites (locations) where vegetation is present or has the potential to grow. Each site $i$, with respective location $r_i$, undergoes a Markov jump process with possible states $F$ (forest), $G$ (grass), $B$ (burning), and $A$ (ash). We thus refer to models with this state space as the ``FGBA model'' for brevity. The transition rates between these states depend on the spatial locations and states of all the other sites at locations $\{r_j \in \Omega \, : \, j = 1,\dots,N, j \neq i\}$. This allows us to model both spatial spreading processes (e.g. forest and fire spreading) as well as non-spatial spontaneous transitions (e.g. fires ignited due to lightning strikes or human activity~\cite{archibald2012evolution}). The spatial locations of the vegetation can be drawn from any sufficiently regular distribution on $\Omega$, but, in this paper, we draw from a uniform distribution on $\Omega$ to reflect a homogeneous domain scenario. In the limit as $N \to \infty$, we completely fill the spatial domain with possible sites for vegetation. We can then obtain an approximate continuum model that can be studied as either a spatial mean-field model (via a system of integro-differential equations) or as a non-spatial mean-field model that presents as a system of ordinary differential equations (see section \ref{sec.mean_field} below for details).

We make the following assumptions about the possible transitions in our model:
\begin{enumerate}
    \item Forest trees can expand into nearby grass and ash due to local spreading of seeds,
    \item burning sites can expand into nearby forest and grass due to local fire spread,
    \item Burning sites are spontaneously quenched into ash at a fixed rate,
    \item Grass can regrow spontaneously from ash at a fixed rate due to homogeneous dispersal of grass seeds,
    \item Forest can spontaneously transition to grass at a fixed rate due to non-fire related mortality.
\end{enumerate}
Mathematically, we represent these assumptions by allowing each site to transition between the states $F$, $G$, $B$, $A$ at exponentially distributed times with rates given by the sum of spatial spreading and spontaneous processes (see Table~\ref{tab:neighbor}). The parameters in the transition rates are defined as follows:
\begin{itemize}
    \item $\varphi_G, \varphi_A$ are constants controlling the rate of forest seeding into grass and into ash, respectively.
    \item $\beta_F, \beta_G$ are constants controlling the rate of fire spread within forest and within grass, respectively.
    \item $W_F, W_B: \mathbb{R}^+ \rightarrow \mathbb{R}^+$ are forest spread and burning spread kernels, respectively, which control the extent of spatial spreading as a function of distance between the interacting sites. We assume that the kernels are $\mathcal{C}^{\infty}$.
\end{itemize}

\renewcommand{\arraystretch}{1.5}
\begin{table}[ht!]
\centering
\begin{tabular}{  c | c | c  } 
  \hline
  Transition & Transition rate at site $i$  & Ecological process \\
  \hline
  $G \rightarrow F$ & $ \frac{ \varphi_G}{N} \sum_{j=1}^NW_F(r_i-r_j) \mathbbm{1}_{\{X^j(t)=F\}}$ &  forest spreading into grass\\ 
  \hline
  $A \rightarrow F$ & $\frac{\varphi_A}{N} \sum_{j=1}^NW_F(r_i-r_j) \mathbbm{1}_{\{X^j(t)=F\}}$ & forest spreading into ash \\ 
  \hline
  $F \rightarrow B$ & $\frac{\beta_F}{N}\sum_{j=1}^NW_B(r_i-r_j)\mathbbm{1}_{\{X^j(t)=B\}}$ &  fire spreading into forest \\ 
  \hline
  $G \rightarrow B$ & $\frac{\beta_G}{N}\sum_{j=1}^NW_B(r_i-r_j)\mathbbm{1}_{\{X^j(t)=B\}}$ &  fire spreading into grass \\ 
  \hline
  $F \rightarrow G$ & $\mu$ & non-fire forest mortality \\ 
  \hline
  $B \rightarrow A$ & $q$ & fire quenching \\ 
  \hline
  $A \rightarrow G$ & $\gamma$ & grass regrowth from ash \\ 
  \hline 
\end{tabular} 
\caption{Transition rates between states and their corresponding ecological processes.}
\label{tab:neighbor}
\end{table}

For the neighbor (local) transition rate parameters for fire and forest invasion ($\varphi_G$, $\varphi_A$, $\beta_F$, $\beta_G$), the subscript is the first English letter of the state before the relevant transition and the main letter is the Greek letter of the state after the relevant transition. We use Greek letters without subscripts for spontaneous transition rates ($\mu, q, \gamma$). We normalize the neighbor spread rates by $\frac{\text{area}(\Omega)}{N}$ to ensure that the spreading rates remain bounded in the limit $N \rightarrow \infty$. 

In addition to the diffusive mode of fire spread described by the $W_B(\cdot)$ kernel and $\beta_G, \beta_F$ parameters, our modeling framework also includes a ``cascade'' mode of vegetation burning motivated by percolation theory. Assuming that fire spread through the grass-forest land is well-approximated as a percolation process, the flammability of any grass site will increase markedly once the fraction of grass sites in the nearby vicinity of the grass site surpasses the critical percolation threshold~\cite{schertzer2015implications}; this mechanism is also well supported by empirical data~\cite{archibald2009limits,pueyo2010testing,van2018fire}. We incorporate this effect by adding the terms
\begin{align}
    \Phi_G\bigg(\frac{1}{N}\sum_{i=1}^N W_G(r_i -r_j)\mathbbm{1}_{\{X^j(t)=G\}}\bigg) \quad \text{ and } \quad \Phi_F\bigg(\frac{1}{N}\sum_{i=1}^N W_G(r_i -r_j)\mathbbm{1}_{\{X^j(t)=G\}}\bigg)\nonumber
\end{align}
to the $G \rightarrow B$ and $F \rightarrow B$ transition rates of site $i$ at time $t$. Here $W_G: \mathbb{R}^+ \rightarrow \mathbb{R}^+$ is another smooth spreading kernel and $\Phi_G, \Phi_F:\mathbb{R}^+ \rightarrow \mathbb{R}^+$ are smooth sigmoidal functions which output flammability as a function of the local grass cover and the kernel $W_G$. More explicitly, we assume that $\Phi_G(\cdot)$ and $\Phi_F(\cdot)$ will take the form
\begin{align*}
    \Phi_G(x) = g_0 + \frac{g_1 - g_0}{1 + e^{-(x-\theta_G)/s_G}} \quad \text{ and } \quad \Phi_F(x) = f_0 + \frac{f_1 - f_0}{1 + e^{-(x-\theta_F)/s_F}}
\end{align*}    
where the parameters are defined as follows: 
\begin{itemize}
    \item $f_0, g_0$ are the baseline spontaneous flammability of forest or grass sites when no other grass sites are present, e.g. ignitions due to lightning strikes or human activity~\cite{archibald2012evolution}.
    \item $f_1, g_1$ are the total flammability of a grass or forest site, respectively in the limit when the land patch is purely grassland.
    \item $\theta_{F}, \theta_G$ are the percolation thresholds in forest and grass, respectively (we will use $\theta = \theta_F = \theta_G \approx 0.4$ as used for percolation in a square lattice~\cite{gebele1984site}).
    \item $s_F, s_G$ are non-negative constants controlling the width of forest and grass sigmoids, respectively. The sigmoids are assumed to be smooth approximations of increasing step functions, so we set $s_F = s_G = 0.05$.
\end{itemize}

The FGBA model outlined above spans multiple distinct timescales. Fire dynamics occur on a time scale of hours, grass regrowth on a time scale of months, and forest dynamics occur on a time scale of decades. In particular, we will estimate the parameter values from the expected time between events. Since we assume the transitions follow an exponential distribution, the rates are the inverse of the expected time between events. The rate parameters all have the same units of yr$^{-1}$, and we assume that the vegetation sites are separated by an average distance of $\SI{10}{\meter}$. 

The vegetation sites are chosen randomly within a compact square domain $[0, L] \times [0, L] \subset \mathbb{R}^2$ for some fixed $L >0$. We use periodic boundary conditions to reduce boundary effects and to model an infinite domain. The distance between two points $r, s \in [0, L] \times [0,L]$ is then computed as
\begin{align*}
    |r - s|^2 =  \frac{L}{2\pi} \Big( \Big[ \text{arg}\Big(e^{i \frac{2 \pi (r_x - s_x)}{L}}\Big)  \Big]^2 + \Big[ \text{arg}\Big(e^{i \frac{2 \pi (r_y - s_y)}{L}}\Big) \Big]^2 \Big)
\end{align*}
where the argument function has range $[0, 2\pi)$. 

We use Gaussian functions for the spreading kernels $W_G(\cdot)$, $W_F(\cdot)$ and $W_B(\cdot)$, although in principle, different functions could be used to model alternative spreading mechanisms if desired (such as very long range seed dispersal modeled by heavy-tailed kernels). A vegetation site located at $r_i \in \Omega$ has spreading kernels
\begin{align*}
    W_{\square}(r_j, r_i) &= \frac{1}{2 \pi \sigma_{\square}^2}\text{exp}\Big(-\frac{|r_j - r_i|^2}{2\sigma_{\square}^2}\Big)
\end{align*}
where $\square = G, F,$ or $B$. To preserve the physical meaning of the three parameters $\sigma_{\square}$ under changes in $L$ and/or $N$, we choose the values of $\sigma_{\square}$ to be in units of the average spacing between the vegetation sites in the domain which we define as $\Delta x \equiv \frac{L}{\sqrt{N}}$. The normalization constant of $W_{\square}$ was chosen that 
\begin{align}
    \lim_{N \rightarrow \infty} \frac{1}{N}\sum_{i=1}^NW_{\square}(r_j, r_i) &= \int_{\Omega} W_{\square}(r_j, r_i) dr_i \approx \int_{\mathbb{R}^2} W_{\square}(r_j, r_i) dr_i = 1
    \label{eqn:W_normalization}
\end{align}
where the approximation is valid when $\sigma_{\square} \ll L$. 

The model described in this section is similar to previous models proposed by Hebert-Dufresne et al. ~\cite{hebert2018edge} and Wuyts and Sieber~\cite{wuyts2023emergent}, particularly the choice of the state space, but has several distinguishing features. Firstly, our model allows the sites to be distributed arbitrarily in $\Omega$, while the related models use a square lattice of sites. Secondly,~\cite{hebert2018edge} and~\cite{wuyts2023emergent} include additional spontaneous transitions $G, A, \rightarrow F$ to model the homogeneous long-distance dispersal of forest seeds. Our model more realistically incorporates long-distance forest seed dispersal by allowing an appropriately heavy-tailed kernel for $W_F(\cdot)$. Lastly, the related models assume that spreading processes occur strictly via nearest neighbor spreading across adjacent lattice sites, while our model allows for more general spreading mechanisms via heavy-tailed kernels often used for seed dispersal~\cite{nathan2012dispersal}. Moreover, as we show below, our spatial stochastic model has immediate, rigorous connections with both spatial and non-spatial mean-field versions, which greatly assist with analysis and interpretability (see section \ref{sec.mean_field}).

\subsubsection{Parameter estimates}
In this section, we set reasonable ranges and approximate values for all parameters in the FGBA model. We assume that $L$ represents the physical side length of the area of land under study. Then, the average spacing between vegetation sites, assuming a uniform probability distribution, is $\Delta x \equiv \frac{L}{\sqrt{N}}$. We now estimate the parameter values by considering the case where $\Delta x = 1$ m and $L = 100$ m are fixed, requiring $N = 10^4$. 

\begin{itemize}
    \item A square patch of grassland of side length $\Delta x$ surrounded by fire will burn after several minutes. Using Table~\ref{tab:neighbor} and Eq.~(\ref{eqn:W_normalization}), the burning rate can be approximated as $\beta_G$ then $\beta_G \approx 10^5$ yr$^{-1}$.
    \item A square patch of forest of side length $\Delta x$, which is surrounded by fire, will burn after about an hour. The burning rate can be approximated as $\beta_F$ so $\beta_F \approx 10^4$ yr$^{-1}$.
    \item A burning site is expected to burn for several hours  before turning into ash, i.e. $q \approx 10^3$ yr$^{-1}$
    \item An ash site is expected to regrow grass in several months i.e. $\gamma \approx 10^1$ yr$^{-1}$.
    \item A square patch of grass or ash of side length $\Delta x$ surrounded by forest will grow a forest tree after several years to a decade i.e. $\varphi_A \approx \varphi_G \approx 10^{0}$ to $10^{-1}$ yr$^{-1}$.
    \item A tree at a square forest site of side length $\Delta x$ will spontaneously die from non-fire related causes after about $100$ years i.e. $\mu \approx 10^{-2}$ yr$^{-1}$.
    \item A square grass or forest vegetation site of side length $\Delta x$ completely surrounded by forest will spontaneously catch fire once every $100$ years i.e. $f_0, g_0 \approx 10^{-2}$.
    \item A grass or forest vegetation site will have increased flammability when completely surrounded by grass compared to a forest. We set $f_0 = 10^{-1}$ yr$^{-1}$ and $g_0 = 10^0$ yr$^{-1}$.
    \item The standard deviation of a fire patch spreading throughout grass or forest vegetation 
    is approximately $5$ m, so set $\sigma_F \approx 5$ m. 
    \item The standard deviation of a forest spreading throughout grass or ash is approximately $5$ m, so set $\sigma_B \approx 5$ m. 
    \item The flammability of a vegetation site is dependent on grass proportions in a $50$ m radius up to one standard deviation, so $\sigma_G \approx 50$ m. 
\end{itemize}
We summarize the parameter values, and a representative fixed value is given below in Table~\ref{tab:values_time_sep}. In this table, $G(t)$ represents the fraction of the landscape occupied by grass sites.

\begin{table}[ht!]
\begin{tabular}{ c | c | c | c } 
  \hline
  Parameter & Value & Units & Ecological Interpretation  \\ 
  \hline
  $\beta_G$ & $10^5$ & yr$^{-1}$ & magnitude of fire spread rate over grass \\
  \hline 
  $\beta_F$ & $10^4$ & yr$^{-1}$ & magnitude of fire spread rate over forest \\
  \hline
  $q$ & $10^4$ & yr$^{-1}$ & rate at which fire is quenched \\
  \hline
  $\gamma$ & $10^2$ & yr$^{-1}$ & rate at which grass regrows from ash \\
  \hline 
  $\varphi_A$ & $1$ & yr$^{-1}$ & magnitude of forest spread rate over ash \\
  \hline
  $\varphi_G$ & $1$ & yr$^{-1}$ & magnitude of forest spread rate over grass \\
  \hline 
  $\mu$ & $10^{-2}$ & yr$^{-1}$ & tree mortality rate due to non-fire-related causes \\
  \hline 
  $f_0$ & $10^{-2}$ & yr$^{-1}$ & spontaneous fire rate for each forest site when $G(t) \approx 0$ \\
  \hline 
  $f_1$ & $10^{-1}$ & yr$^{-1}$ & spontaneous fire rate for each forest site when $G(t) \approx 1$ \\
  \hline
  $g_0$ & $10^{-2}$ & yr$^{-1}$ & spontaneous fire rate for each grass site when $G(t) \approx 0$ \\
  \hline 
  $g_1$ & $10^{-1}$ & yr$^{-1}$ & spontaneous fire rate for each grass site when $G(t) \approx 1$ \\
  \hline
  $\sigma_F$ & $5$ & m & width of forest spreading kernel\\
  \hline
  $\sigma_B$ & $5$ & m & width of fire spreading kernel\\
  \hline
  $\sigma_G$ & $50$ & m & width of spontaneous fire kernel\\
  \hline
\end{tabular}
\caption{Estimates and representative values of the parameters for the FGBA model in units of yr$^{-1}$} \label{tab:values_time_sep}
\end{table}

There are three levels of time separation in the rate parameters given by $\beta_G,\beta_F, q >  \gamma \gg \varphi_A, \varphi_G, \mu, f_0, g_0$. This concludes our mathematical description of the spatial FGBA model. We summarize the model in the state transition diagram in Fig.~(\ref{fig:spatial_FGBA}), showing the relevant parameters and functions governing transitions between states at each vegetation site. 
\begin{figure}[ht]
	\centering
	\tikzset{forest/.style={circle, thick, minimum size=1.2cm, draw=medgreen!60, fill=medgreen!30},
		grass/.style={circle, thick, minimum size=1.2cm, draw=green!45, fill=green!15},
		burning/.style={circle, thick, minimum size=1.2cm, draw=orange!50, fill=orange!20},
        ash/.style={circle, thick, minimum size=1.2cm, draw=gray!50, fill=gray!20},
		arrow/.style={thick,->,>=stealth}}
	\begin{tikzpicture}[>=latex,text height=1.5ex,text depth=0.25ex]
	\matrix[row sep=1cm,column sep=1cm] (mat) {
		&
		\node (F) [forest] {$F$}; &
		&
		\node (G)   [grass] {$G$}; \\
		&
		\node (A) [ash] {$A$}; &
		&
		\node (B)   [burning] {$B$}; \\
	};
\begin{scope}[arrow]
\draw[medgreen] (F.10) -- (G.170) node[midway, above]{$\mu$};
\draw[orange] (F) -- (B) node[midway, left]{
    $\Phi_F, \beta_F$};
\draw[orange] (A.30) -- (G.225) node[midway, below]{$\gamma$};
\draw[orange] (B) -- (A) node[midway, below]{$q$};
\draw[orange] (G) -- (B) node[midway, right]{$\Phi_G,\beta_G$};
\draw[medgreen] (A) -- (F) node[midway, left]{$\varphi_A$};
\draw[medgreen] (G.190) -- (F.350) node[midway, below]{$\varphi_G$};
\end{scope}
	\end{tikzpicture}
	\caption{State transition diagram of the spatial FGBA model. Transition arrows are labeled with the relevant parameters and/or flammability functions. Forest and grass/fire timescale transitions are shown in green and orange, respectively.}
 \label{fig:spatial_FGBA}
\end{figure}
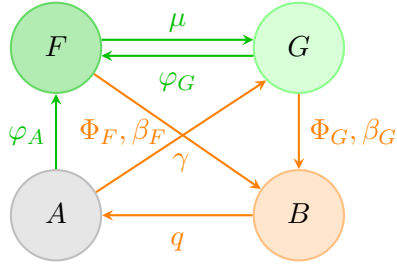

\subsection{Mean-Field Approximations of the Spatial FGBA Model}\label{sec.mean_field}
Theorem 2.2.1 in~\cite{patterson2020probabilistic} shows that the spatial FGBA model with uniformly randomly distributed sites on $\Omega$ converges to a spatial McKean-Vlasov process in the limit as $N$ (the number of sites) tends to infinity. We may think of this limiting processes as a type of spatial mean-field limit since it retains information about spatial structure, but the vegetation at each location only interacts with other locations through weighted spatial averages of all other locations. Let $X(r,t)$ indicate the probability that the limiting McKean-Vlasov process at location $r$ and time $t$ is in state $X$. Then the forward equation governing the dynamics of the probability densities of the limiting process are given by the following system of integro-differential equations (IDEs):
\begin{align}\label{eq.spatial_MF}
    \frac{d}{dt} F(r,t) &= (\varphi_G G(r,t)  + \varphi_A A(r,t))\int_{\Omega}W_F(r-r')F(r',t)dr' - \mu F(r,t)\nonumber\\
    &-\beta_F F(r,t) \int_{\Omega}W_B(r-r')B(r',t)dr' -\Phi_F\bigg[\int_{\Omega}W_G(r-r')G(r',t)dr'\bigg] F(r,t),\nonumber\\
    \frac{d}{dt}G(r,t) &= \gamma A(r,t) - \varphi_G G(r,t) \int_{\Omega}W_F(r-r')F(r',t)dr' +\mu F(r,t) \nonumber\\
    &-\Phi_G\bigg[\int_{\Omega}W_G(r-r')G(r',t)dr'\bigg] G(r,t) - \beta_G G(r,t) \int_{\Omega}W_B(r-r')B(r',t)dr',\nonumber\\
    \frac{d}{dt}B(r,t) &= ( \Phi_F F(r,t) + \Phi_G G(r,t) )\bigg[\int_{\Omega}W_G(r-r')G(r',t)dr'\bigg] - qB(r,t) \\
    &+ (\beta_G G(r,t) + \beta_F F(r,t)) \int_{\Omega}W_B(r-r')B(r',t)dr',\nonumber\\
    \frac{d}{dt} A(r,t) &= q B(r,t) - \gamma A(r,t) - \varphi_A A(r,t) \int_{\Omega}W_F(r-r')F(r',t) dr',\nonumber
\end{align}
for each $r \in \Omega$. Moreover, $F(r,t)+G(r,t)+B(r,t)+A(r,t) = 1$ for each $(r,t)\in \Omega\times \mathbb{R}_+$. 

To gain a rough, intuitive understanding of the stochastic FGBA model's dynamics, we will analyze a more analytically tractable nonspatial mean-field model derived from \eqref{eq.spatial_MF} that takes the form of a system of ordinary differential equations. The general structure of the solution space of the spatial model \eqref{eq.spatial_MF} is not analytically tractable. One can investigate the linear stability of spatially homogeneous steady-state solutions, but we expect similar stability criteria to the nonspatial mean-field case for most reasonable choices of the kernel functions, as has been shown for related forest-savanna models posed in this framework~\cite{patterson2024pattern}. Using the nonspatial mean-field approximation results as a baseline will also allow us to isolate the impact of spatial structure in the stochastic FGBA model. In addition to spatial structure, the stochastic FGBA model allows us to retain stochastic, discrete and finite-size effects, which we expect to be especially important for short-time and transient dynamics~\cite{hastings2018transient}, which we can contrast against the nonspatial mean-field model.

To obtain a nonspatial (spatially implicit) mean-field model from \eqref{eq.spatial_MF}, we assume that the types of states are well-mixed throughout the spatial domain $\Omega$ so that interactions depend only on the fractions of land occupied by each state. One way to achieve this is to take the kernels functions to be uniform on $\Omega$ and then define the land cover proportions $F(t) := \int_\Omega F(r',t)dr'/\text{area}(\Omega)$, $G(t) := \int_\Omega G(r',t)dr'/\text{area}(\Omega)$, and so on. The system \eqref{eq.spatial_MF} then simplifies to the system of ordinary differential equations:
\begin{align}
    \begin{cases}
    \frac{d}{dt}F(t) = (\varphi_G G + \varphi_A A) F -\beta_F B F -\Phi_F(G)F- \mu F \\
    \frac{d}{dt} G(t) = \gamma A +\mu F- \varphi_GFG -\Phi_G(G) G - \beta_GBG \\
    \frac{d}{dt} B(t) =\Phi_G(G) G + \Phi_F(G)F + (\beta_G G + \beta_F F)B -q B\\
    \frac{d}{dt} A(t) = q B - \gamma A - \varphi_AF A. \label{eqn:mf}
    \end{cases}
\end{align}

The state transition diagram for the spatial FGBA model in the mean field limit is illustrated in Fig.~(\ref{fig:markov_chain}). First notice that due to the $G \xrightarrow{\Phi(G)G} B \xrightarrow{qB} A \xrightarrow{\gamma A} G$ cycle then if any one of $\overline{G}, \overline{B}$, or $ \overline{A}$ is nonzero then all three must be nonzero. Next, since there is an $F \xrightarrow{\mu F}G$ transition then $\overline{F} > 0$ implies $\overline{G} > 0$. Thus, we can classify all steady states as either GBA (where $\overline{F} = 0$ and $\overline{G}, \overline{B}, \overline{A} \neq 0$) or FGBA (where $\overline{F}, \overline{G}, \overline{B}, \overline{A} \neq 0$). 
\begin{figure}[ht]
	\centering
	\tikzset{forest/.style={circle, thick, minimum size=1.2cm, draw=medgreen!60, fill=medgreen!30},
		grass/.style={circle, thick, minimum size=1.2cm, draw=green!45, fill=green!15},
		burning/.style={circle, thick, minimum size=1.2cm, draw=orange!50, fill=orange!20},
        ash/.style={circle, thick, minimum size=1.2cm, draw=gray!50, fill=gray!20},
		arrow/.style={thick,->,>=stealth}}
	\begin{tikzpicture}[>=latex,text height=1.5ex,text depth=0.5ex]
	\matrix[row sep=1.25cm,column sep=2.6cm] (mat) {
		&
		\node (F) [forest] {$F$}; &
		&
		\node (G)   [grass] {$G$}; \\
		&
		\node (A) [ash] {$A$}; &
		&
		\node (B)   [burning] {$B$}; \\
	};
\begin{scope}[arrow]
\draw[medgreen] (F.10) -- (G.170) node[midway, above]{$\mu F$};
\draw[orange] (F) -- (B.130) node[midway, left]{
    $\Phi_F(G)F + \beta_F BF\quad$};
\draw[orange] (A.30) -- (G.220) node[midway, below]{$\gamma A$};
\draw[orange] (B) -- (A) node[midway, below]{$qB$};
\draw[orange] (G) -- (B) node[midway, right]{$\Phi_G(G)G + \beta_GBG$};
\draw[medgreen] (A) -- (F) node[midway, left]{$\varphi_AFA$};
\draw[medgreen] (G.190) -- (F.350) node[midway, below]{$\varphi_GFG$};
\end{scope}
	\end{tikzpicture}
	\caption{FGBA model state transition diagram in the mean-field limit. Transition arrows are labeled by the transition rates. Forest and fire timescale transitions are shown in green and orange, respectively.}
 \label{fig:markov_chain}
\end{figure}
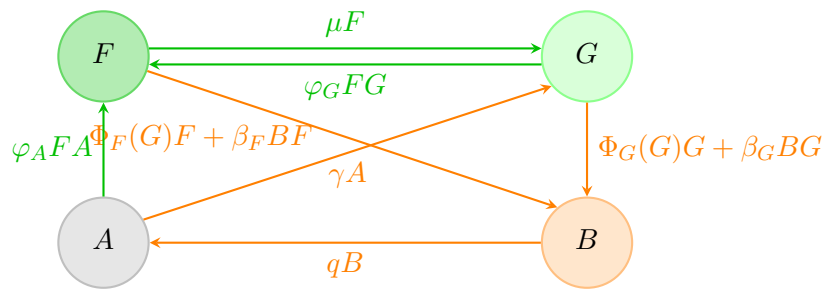

\section{Analysis of the nonspatial FGBA model}\label{sec.nonspatial_model}
Before studying the full spatial stochastic FGBA model introduced above, we perform some qualitative stability analysis of solutions to the nonspatial mean-field version of the model given by (\ref{eqn:mf}).
\subsection{GBA Steady States}
As it is the simplest subsystem of the model, we first consider the case when forest is absent from the landscape. Set $F = 0$ and use the relation $A = 1 - G - B$ to eliminate $A$ from (\ref{eqn:mf}) to reduce the system (\ref{eqn:mf}) to
\begin{align}
    \begin{cases} 
        \frac{d}{dt}G(t) &= \gamma (1-G-B) - \Phi_G(G) G- \beta_GBG \\
    \frac{d}{dt} B(t)  &= \Phi_G(G)G + \beta_G GB - qB.
    \end{cases} \label{eqn:GB}
\end{align} 
Establishing the forward invariance of the GBA subspace is straightforward, and so the proof of the following proposition is deferred to Appendix \ref{secA1}.

\begin{proposition}
For any ecologically relevant initial value (i.e. $G(0) \geq 0$, $B(0) \geq 0$, $G(0) + B(0) \leq 1$), the solution $(G(t), B(t))$ for Eqs.~(\ref{eqn:GB}) remains ecologically relevant for all times $t \in \mathbb{R}$.
\end{proposition}

Next, we solve for the steady-state equilibria of the system by setting both equations in System~(\ref{eqn:GB}) equal to $0$. Solving the system then gives the condition
\begin{align}
    \Phi_G(\overline{G}) &= (1-\overline{G})\frac{\gamma q}{\gamma + q} \bigg(\frac{1}{\overline{G}}-\frac{\beta_G}{q}\bigg) \nonumber
\end{align}
for a steady state (note that division by $\overline{G}$ is a valid operation since $\overline{G}>0$). Unfortunately, closed-form expressions of the equilibria cannot be obtained except in special cases due to the sigmoidal properties of the function $\Phi_G$, but we can readily deduce the uniqueness of the steady state: 

\begin{proposition}\label{prop_2}
For any positive parameter values $\gamma, q,$ and $\beta_G$ there exists a unique GBA steady state. 
\end{proposition}
The proof of Proposition \ref{prop_2} is deferred to Appendix \ref{secA1} in the interests of brevity. 

\subsubsection{Stability of the GBA Steady State}
In this section, we demonstrate two results regarding the stability of the GBA steady state. The first concerns stability within the GBA subsystem, and the second gives an interpretable condition characterizing when we can expect the GBA subsystem to be stable to forest invasion.
\begin{theorem}
    The GBA steady state is stable to perturbations within the $F=0$ subspace for any parameter values. 
\end{theorem}
\begin{proof}
    We compute the linearization matrix at $F = 0$ of System~\ref{eqn:GB} which gives
    \begin{align*}
        \textbf{J}(F = 0, G,B) = \begin{bmatrix} -\gamma-\Phi_G'(G)G-\Phi_G(G)-\beta_G B & -\gamma - \beta_G G \\
        \Phi_G'(G) G + \Phi_G(G) + \beta_G B & -q + \beta_G G
    \end{bmatrix}.
    \end{align*}
    The fixed point $(\overline{G}, \overline{B})$ is stable to perturbations within the $F = 0$ boundary exactly when $J_{F=0}(\overline{G}, \overline{B})$ has eigenvalues with strictly negative real parts. By the Routh-Hurwitz stability criteria this occurs if and only if $\Tr(\textbf{J}_{F=0}(\overline{G}, \overline{B})) < 0$ and $\det(\textbf{J}_{F=0}(\overline{G}, \overline{B})) >0$~\cite{edelstein2005mathematical}. Note that since $\overline{G} < \frac{q}{\beta_G}$ then $-q + \beta_G \overline{G} < 0$ so $\Tr(J_{F=0}(\overline{G}, \overline{B})) < 0$ is clearly satisfied. The condition $\det(\textbf{J}_{F=0}(\overline{G}, \overline{B})) >0$ is equivalent to 
    \begin{align*}
        \Phi'(\overline{G}) > \frac{\gamma q}{(\gamma + q)^2}\bigg(\frac{\beta}{q}-\frac{1}{\overline{G}^2} \bigg) = \mathcal{F}'(\overline{G})
    \end{align*}
    which clearly holds from graphical inspection or more rigorously by recalling from our earlier argument that $\Phi_G'(\overline{G})  \geq 0$ while $\mathcal{F}'(\overline{G}) < 0$. 
\end{proof}

\begin{proposition}
    The GBA steady state is stable to invasions by $F$ exactly when 
    \begin{align}
        \beta_F \overline{B} + \Phi_F(\overline{G}) + \mu > \varphi_G \overline{G} + \varphi_A \overline{A}. \label{eqn:GBA_stability}
    \end{align}
\end{proposition}

\begin{proof}
    We compute the linearization matrix for the full FGBA system and evaluate it at $F = 0$. We use $A = 1 - F - G - B$ to obtain a dynamical system in the three variables $F$, $G$, and $B$ which gives the matrix
    \begin{align*}
        &\textbf{J}(F = 0, G,B) \\
        &= 
        \begin{bmatrix} 
        \varphi_G G + \varphi_A A - \beta_F B -\Phi_F(G) - \mu & 0 & 0  \\
        -\gamma-\varphi_G G +\mu & -\gamma-\Phi_G'(G)G-\Phi_G(G)-\beta_G B & -\gamma - \beta_G G \\
        \Phi_F(G) + \beta_F B & \Phi_G'(G) G + \Phi_G(G) + \beta_G B & -q + \beta_G G
        \end{bmatrix}.
    \end{align*}
    The GBA steady state is stable to invasion by $F$ if all the eigenvalues of $\textbf{J}(\overline{F} = 0, \overline{G},\overline{B})$ have negative real parts. We have already shown that the eigenvalues of the $2 \times 2$ submatrix in the lower right have strictly negative real parts so the GBA steady state is stable if and only if the entry in the upper left is negative.
\end{proof}

An intuitive interpretation of equation~(\ref{eqn:GBA_stability}) is that the GBA equilibria is resistant to invasion by forest when the rate of tree mortality by fire spread through forest, spontaneous forest fires, and natural mortality (left-hand side) exceeds the rate of forest seeding into the steady state grass and ash land (right-hand side). Since a unique GBA steady state exists for any value of the system parameters, but its stability can vary based on the parameter values, we expect to find transcritical bifurcations within the parameter space when equality occurs in equation ~(\ref{eqn:GBA_stability}). In other words, a stable GBA steady state bifurcates into an unstable no-forest state and a stable forest state.

Some example bifurcation diagrams in $\beta_G$, $\beta_F$, $\varphi = \varphi_G = \varphi_A$ and $\mu$ using partially timescale-separated parameter values are given in Fig.~(\ref{fig:some_time_sep_one_param_bif_FGBA}). As expected, when $\varphi$ is increased past a critical value or when $\beta_G$, $\beta_F$, and $\mu$ are decreased past a critical value, the proportion of forest sites within $\Omega$ becomes non-zero. However, the parameter values produce a mostly ash grassland steady state which is not very ecologically realistic. Nonetheless, it will later be shown that the spatial FGBA model has grassy steady states distinct from the ash-dominated, spatially homogeneous GBA steady state.


\subsubsection{Estimates of the GBA Steady State}
We will now show that using our estimates of the parameter values, the steady state $(\overline{G}, \overline{B}, \overline{A})$ can be easily estimated. 
Notice that $\mathcal{F}(G)$ becomes asymptotic to $\frac{\gamma q}{\gamma + q} (\frac{1}{G} - \frac{\beta_G}{q})$ as $G \rightarrow 0$ and the coefficient $\frac{\gamma q}{\gamma + q}$ is of order $2$. Since $\Phi_G(G)$ is of order $0$ or lower for $G \in [0,1]$ then we can approximate the condition $\Phi_G(G) = \mathcal{F}(G)$ as $0 = \mathcal{F}(G)$ which gives $\overline{G} \approx \frac{q}{\beta_G}$. It follows that to increase the portion of land covered in grass, we must either increase $q$, i.e. increasing the fire quenching rate, or decrease $\beta_G$, i.e. decreasing the rate of fire spread through grass, both of which make sense intuitively. By solving System~\ref{eqn:GB} and using $G+B+A=1$ we can compute an estimate of the GBA steady state as well as the ratio $\overline{B}/\overline{A}$:
\begin{align}
    (\overline{G}, \overline{B}, \overline{A}) \approx \bigg(\frac{q}{\beta_G}, \frac{\gamma(1-\frac{q}{\beta_G})}{q+\gamma}, \frac{q(1 - \frac{q}{\beta_G})}{q+\gamma}\bigg) \Longleftrightarrow \frac{\overline{B}}{\overline{A}} \approx \frac{\gamma}{q}.\label{eqn:GBA}
\end{align}
Based on the estimate, the ratio between the portions of land in burning and ash states is determined by $\gamma$, the rate of grass regrowth from ash, and $q$, the rate of fire quenching in a manner that agrees with intuition.

Then Eq.~(\ref{eqn:GBA}) can be used to give Eq.~(\ref{eqn:GBA_stability}) purely in terms of the model parameters:
\begin{align}
     \frac{\beta_F \gamma(1 - \frac{q}{\beta_G})}{q+\gamma}
     + \Phi_F\Big(\frac{q}{\beta_G}\Big)
     + \mu
     > \frac{q\varphi_G}{\beta_G}
     + \frac{q\varphi_A(1-\frac{q}{\beta_G})}{q+\gamma}.
     \nonumber
\end{align}
In particular, if we vary each of the parameters $\beta_F, \gamma, f_0, \mu,$ and $ \varphi = \varphi_A = \varphi_G$ one at a time and fix all other parameters at the estimates given in Table \ref{tab:values_time_sep}, we can compute the transcritical bifurcation points for each parameter. To simplify the calculations, we will treat $\Phi_F$ as a Heaviside step function, that is
\begin{align*}
    \Phi_F(x) = f_0 + (f_1-f_0)H(x - \theta), \quad \text{ where } \quad H(x) = \begin{cases} 1 & x \geq 0, \\ 0 & x < 0, \end{cases}
\end{align*}
which can be viewed as $\Phi_F$ in the limit $s_F \rightarrow 0$.

\subsection{FGBA Model Steady States}
We next solve for the FGBA steady states. We first use the relation $A = 1 - F -G - B$ to eliminate the ash land cover fraction from  (\ref{eqn:mf}), which gives the following system of equations:
\begin{align}
    \frac{d}{dt}F(t) &= (1-F-B) \varphi F -\beta_F B F -\Phi_F(G)F- \mu F \label{eqn:FGBA_F}\\
    \frac{d}{dt}G(t) &= \gamma (1-F-G-B) +\mu F- \varphi FG -\Phi_G(G) G - \beta_GBG  \label{eqn:FGBA_G}\\
    \frac{d}{dt} B(t) &=\Phi_G(G) G + \Phi_F(G)F + (\beta_G G + \beta_F F)B -q B \label{eqn:FGBA_B}.
\end{align}
From this point forward, we will set $\varphi = \varphi_A = \varphi_G$ to simplify calculations. This is an ecologically reasonable assumption since the rate of forest spread into ash should not differ substantially from the rate of forest spread into grassland. Note that a similar assumption was made in~\cite{wuyts2023emergent}. Then after setting $\dot{G}= \dot{B}= \dot{A} = 0$ we can solve for roots $(\overline{F}, \overline{G}, \overline{B}, \overline{A})$. As for the GBA subspace, the full model retains forward invariance of the ecologically relevant region of the phase-space (proof deferred to Appendix \ref{secA1}).

\begin{proposition}\label{prop.fwd.FGBA}
For any ecologically relevant initial value (i.e. $F(0) \geq 0$, $G(0) \geq 0$, $B(0) \geq 0$, $F(0) + G(0) + B(0) \leq 1$), the solution $(F(t), G(t), B(t))$ remains ecologically relevant for all times $t \in \mathbb{R}$.
\end{proposition}

We can now investigate the transcritical bifurcation anticipated to occur at the transition from the GBA steady state to states with non-zero forest. Let $\overline{B}(F)$ be a function that outputs the equilibrium burning proportion(s) $B$ for any given input $F$. The function $\overline{B}(F)$ can be obtained from Eqs.~(\ref{eqn:FGBA_F}) to~(\ref{eqn:FGBA_B}) by solving for $B$ in terms of $F$ after setting $\dot{G}= \dot{B}= 0$ and enforcing $F + G+ B+ A = 1$. In general, $\overline{B}(\cdot)$ is a complicated function we will analyse later. The vector field for forest cover can be expressed as
\begin{align*}
    Q(F) &= (1 - F - \overline{B}(F))F \varphi - \beta_F \overline{B}(F)F - (f_1 + \mu)F \\
    &= -\overline{B}(F) F (\varphi+\beta_F) - \varphi F^2 + F(\varphi -\tilde{\mu}) 
\end{align*}
where we introduce the parameter $\tilde{\mu} := \mu + f_1$ to reduce notational clutter. The normal form for a transcritical bifurcation is $\dot{x}=a_1x+x^2$. To verify that the genericity conditions for a transcritical bifurcation hold, we note first that $Q(F = 0) = 0$ as expected since $F=0$ is an equilibrium. Furthermore,
\begin{align*}
    Q'(F = 0) &= -\overline{B}(0)(\varphi + \beta_F) + \varphi - \tilde{\mu}
\end{align*}
which indicates the stability of the $F = 0$ equilibrium (i.e. the $GBA$ equilibrium). For instance, note that the inequality $Q'(F) < 0$ is exactly the same as Eq.~(\ref{eqn:GBA_stability}) after making the simplifications $\varphi = \varphi_G = \varphi_A$ and $\Phi_F(\overline{G}) = f_1$. Lastly, we have
\begin{align*}
    Q''(0) &= - 2\overline{B}'(0)(\varphi+\beta_F) -2\varphi
\end{align*}
and it follows that as long as $Q''(0) \neq 0$ a non-degenerate transcritical bifurcation occurs at parameter values where $Q'(F=0)$ i.e. when
\begin{align}
    \overline{B}(0) &= \frac{\varphi - \tilde{\mu}}{\varphi+ \beta_F} 
\end{align}
Next, we know that after passing the transcritical bifurcation, the system displays two distinct equilibria, one stable and the other unstable. Since $F=0$ is always an equilibrium, it remains to be determined whether or not the second equilibrium is ecologically plausible, i.e. contained in $\mathcal{F}_{FGBA}$. To do this, we first note that the nonzero forest equilibrium can be calculated from Eq.~(\ref{eqn:FGBA_F}) by dividing out the factor of $F$ and setting $\dot{F} = 0$. This gives an implicit equation
\begin{align*}
    \overline{F}(\overline{B}) &= \frac{(\varphi-\tilde{\mu}) - (\varphi + \beta_F)\overline{B}(\overline{F})}{\varphi}.
\end{align*}
Next we can expand about the equilibria by writing $F = 0 \rightarrow \delta F$ while $\varphi \rightarrow \varphi + \delta \varphi$ or $\beta_F \rightarrow \beta_F + \delta \beta_F$ or $\tilde{\mu} \rightarrow \tilde{\mu}+ \delta \tilde{\mu}$. Then, after keeping only first-order terms, we find that
\begin{align*}
    \bigg[\frac{\varphi}{\varphi + \beta_F} + \overline{B}'(0)\bigg] \delta F = \bigg[ \frac{\tilde{\mu} + \beta_F}{(\varphi + \beta_F)^2}\bigg] \delta \varphi = \bigg[ \frac{\tilde{\mu} - \varphi}{(\varphi + \beta_F)^2}\bigg] \delta \beta_F = \bigg[\frac{-1}{(\varphi + \beta_F)^2}\bigg] \delta \tilde{\mu}.
\end{align*}
Then, the new equilibria emerging at the transcritical point after perturbation of a parameter value is ecologically plausible exactly when $\delta F > 0$. Thus, whether or not the non-zero forest equilibria are ecologically plausible can be determined based on the parameter values $\tilde{\mu}, \beta_F, \varphi$ as well as $\overline{B}'(0)$. 

Next, after some calculations (see Appendix \ref{secA2} for details), we find that the parameter values along which the transcritical bifurcation occurs can be determined to various corrections in orders of $g_1$:
\begin{align}
    \frac{\varphi - \tilde{\mu}}{\varphi+ \beta_F}  = \frac{\gamma (\beta_G - q)}{\beta_G(q+\gamma)} + \frac{q g_1}{\beta_G(\beta_G - q)} + \frac{q(q+\gamma)g_1^2}{(q-\beta_G)^3 \gamma} + \mathcal{O}(g_1^3)
    \label{eqn:BP}
\end{align}
A plot of numerically computed branch points and Eq.~(\ref{eqn:BP}) plotted to various orders in $g_0$ and $f_0$ is given in Fig.~(\ref{fig:BP_approx}). 

\begin{figure}[ht] 
\centering  
\includegraphics[width=0.95\linewidth]{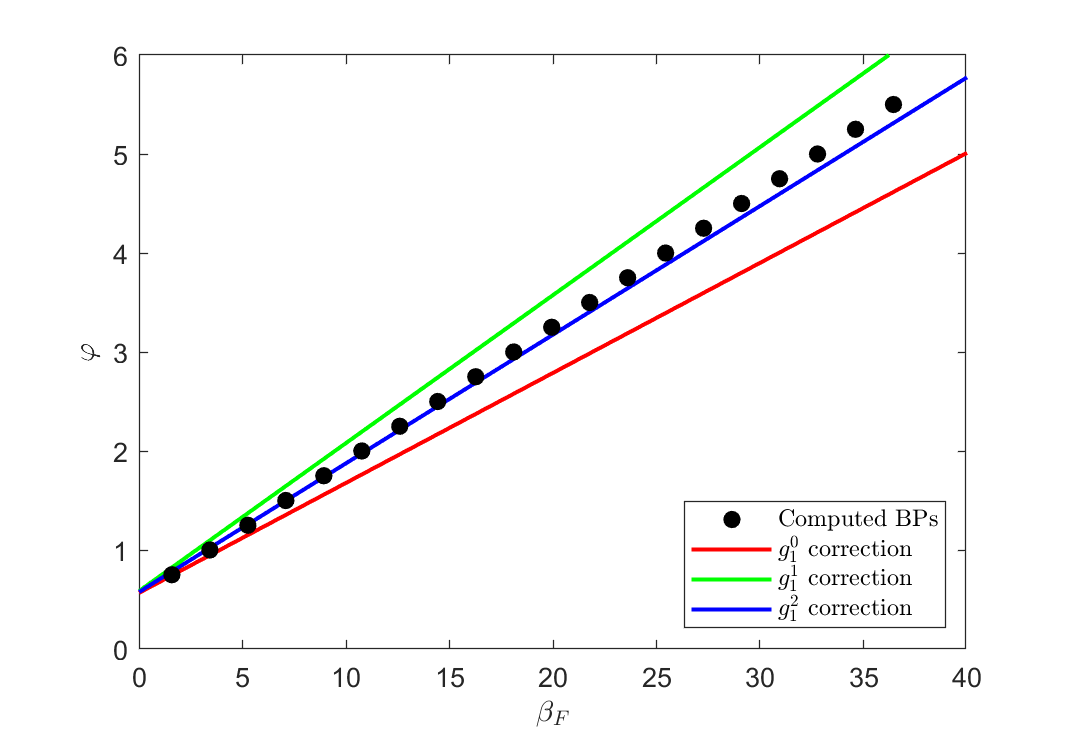}
\caption{A plot of the computed branch points governing the stable to unstable transition for the GBA steady state alongside the analytically predicted branch points in Eq.~(\ref{eqn:BP}). Branch points are plotted as a function of $\beta_F$, the rate of fire spread over forest, and $\varphi$, the rate of forest spreading over grass and ash (assuming $\varphi = \varphi_A = \varphi_G$). Parameter values: $\beta_G = 50$, $\beta_F = 10$, $q = 30$, $\gamma = 10$, $\varphi = 0.1$, $f_0 = g_0 = 0.01$, $f_1 = 0.5$, $g_1 = 1$}
\label{fig:BP_approx}
\end{figure} 

\subsubsection{Bifurcation diagrams (partially timescale separated)}
As the previous section shows, analyzing the FGBA system with a non-zero forest is complex. As such, a numerical study of the system is helpful for a more general understanding of the system's behavior. Unfortunately, numerical computations with timescale-separated parameter values is extremely slow. Thus, we performed a numerical analysis with somewhat less timescale-separated parameter values: $\beta_G = 50$, $\beta_F = 10$, $q = 30$, $\gamma = 10$, $\varphi = 1$, $\mu = 0.01$, $f_0 = g_0 = 0.01$, $f_1 = 0.5$, and $g_1 = 1$.

\begin{figure}[ht] 
\centering  
\includegraphics[width=0.99\linewidth]{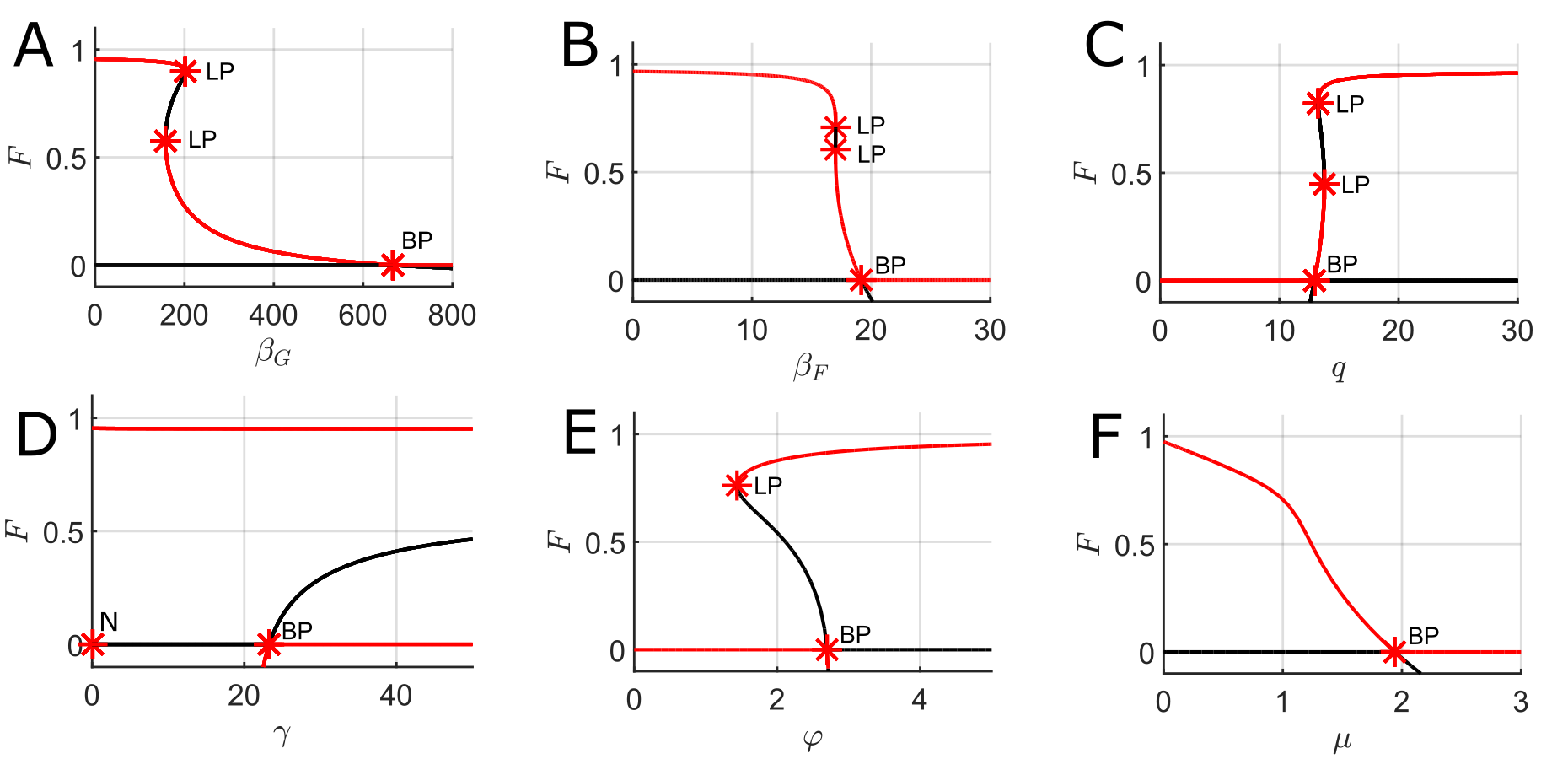}
\caption{One parameter bifurcation diagrams for partially timescale separated parameter values. Red lines indicate stable equilibria, while black lines indicate unstable equilibria. LP indicates a saddle-node bifurcation point (limit point), BP indicates a transcritical bifurcation point (branch point), and N indicates a neutral saddle equilibrium with no stability change.}
\label{fig:some_time_sep_one_param_bif}
\end{figure} 

\begin{figure}[ht] 
\centering  
\includegraphics[width=0.99\linewidth]{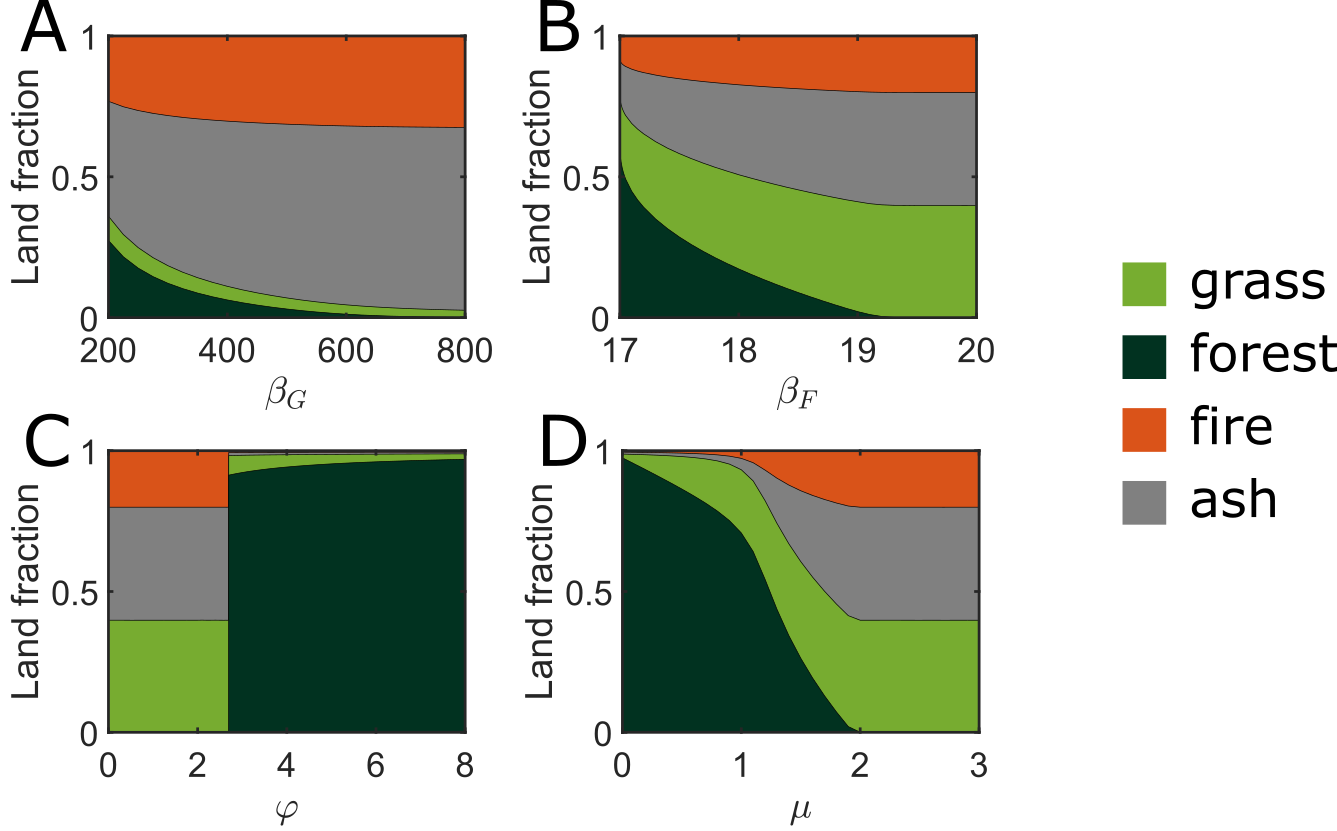}
\caption{Equilibrium land state proportions reached starting from a low forest state ($F=0.01$ initial condition) for partially timescale separated parameter values near branch points.}
\label{fig:some_time_sep_one_param_bif_FGBA}
\end{figure} 

\begin{figure}[h] 
\centering  
\includegraphics[width=0.99\linewidth]{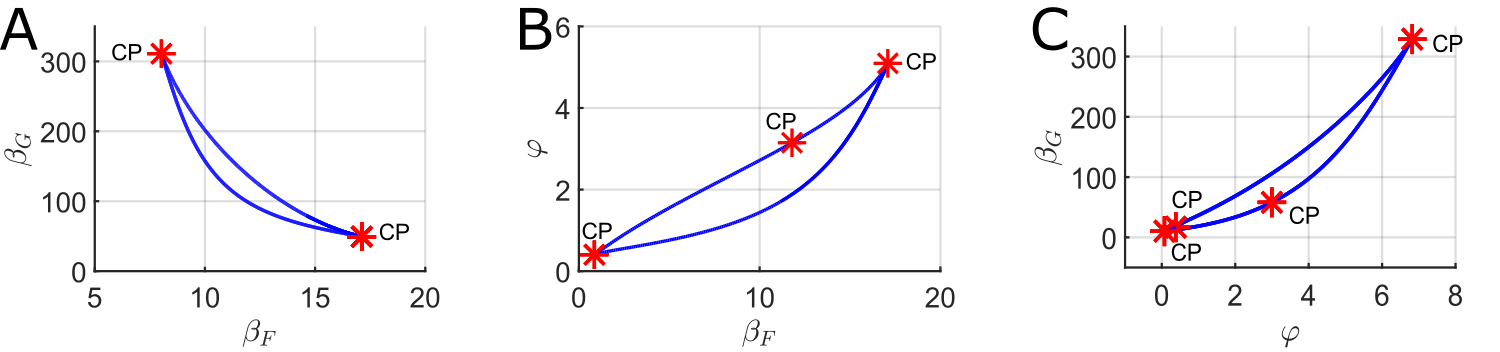}
\caption{Two parameter bifurcation diagrams of the limit points in Fig.~(\ref{fig:some_time_sep_one_param_bif}) for partially timescale separated parameter values. CP indicates a codimension two cusp point where two curves of saddle-node bifurcations intersect.}
\label{fig:some_time_sep_two_param_bif}
\end{figure} 

\begin{figure}[h] 
\centering  
\includegraphics[width=0.95\linewidth]{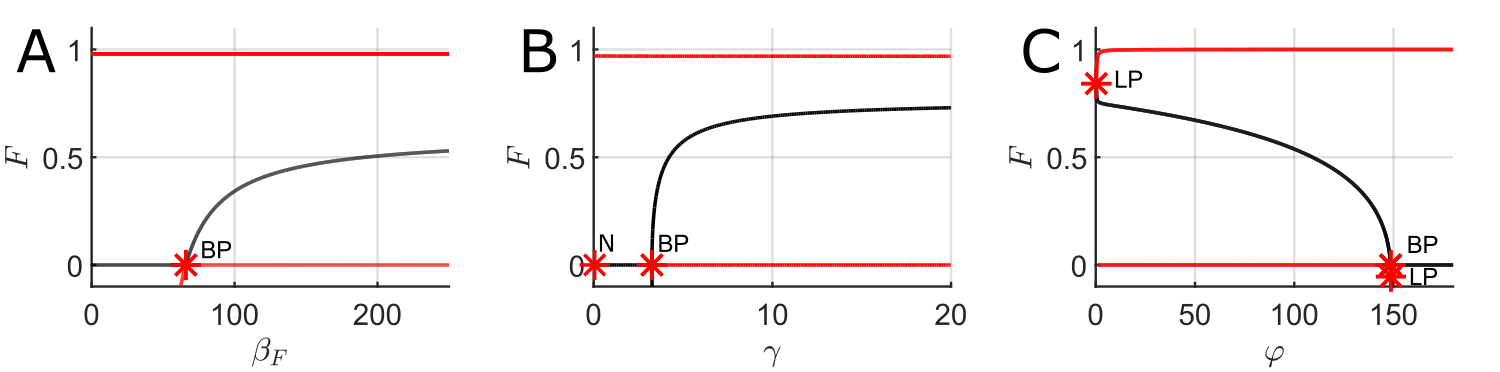}
\caption{One-parameter bifurcation diagrams in $\beta_F$ (fire spreading rate in forest), $\gamma$ (grass regrowth rate from ash) and $\phi$ (fire spreading rate in ash and grass) for fully timescale-separated parameter values. Red lines indicate stable equilibria while black lines indicate unstable equilibria. LP indicates a limit point, BP indicates a branch point, and N indicates a neutral saddle equilibrium.}
\label{fig:full_time_sep_one_param_bif}
\end{figure} 

\begin{figure}[h] 
\centering  
\includegraphics[width=0.95\linewidth]{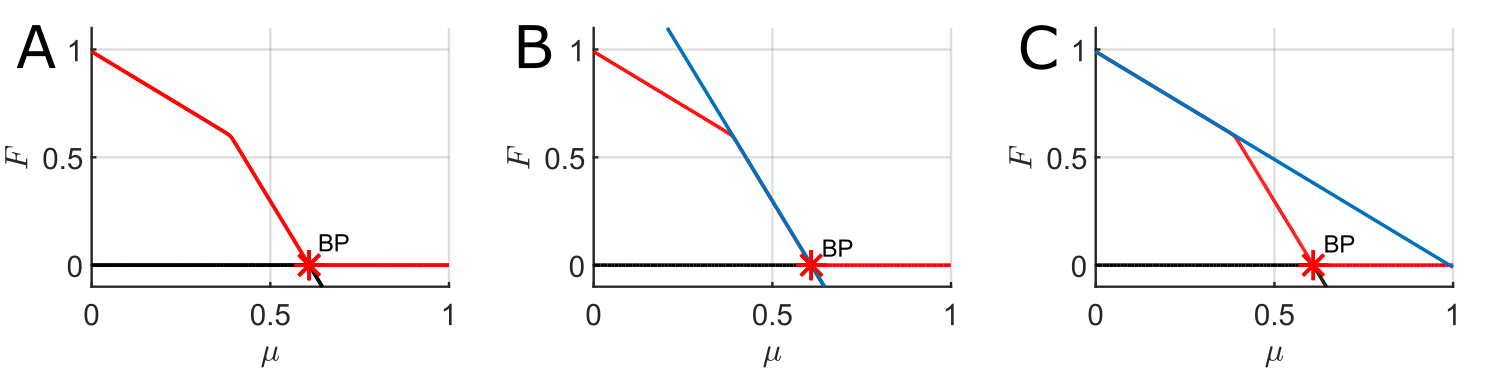}
\caption{A. Bifurcations in $\mu$ (non-fire forest mortality rate) for fully timescale separated parameter values with $\beta_F = 25$. Red lines indicate stable equilibria while black lines indicate unstable equilibria. BP indicates a branch point. The linear regimes near $F=0$ and $F=1$ are plotted in blue in panels B and C, respectively. }
\label{fig:mu_lines}
\end{figure}

As expected, the zero forest cover state is stable at low values of $\varphi$ and high values of $\beta_F$. Notably, a stable high forest state is also present in the same parameter ranges as the low forest states. Further, as expected, a transcritical bifurcation is present at the critical parameter values of both $\varphi$ and $\beta_F$ where the GBA steady state becomes unstable. In the two-parameter bifurcation diagram in Fig.~(\ref{fig:some_time_sep_two_param_bif}), the regions enclosed by the blue saddle-node curves indicate the parameter regimes in which we can expect bistability in forest cover. 

\subsubsection{Bifurcation diagrams (timescale separated)} 
Several one-parameter bifurcation diagrams were plotted using fully timescale-separated parameter values using MatCont~\cite{dhooge2008new}. The values used were $\beta_G = 50000$, $\beta_F = 10000$, $q = 20000$, $\gamma = 500$, $\varphi = 1$, $\mu = 0.01$, $f_0 = g_0 = 0.01$, $f_1 = 0.05$, and $g_1 = 0.1$ for the bifurcation diagrams in Fig.~(\ref{fig:full_time_sep_one_param_bif}). Due to prohibitively slow computation time, two-parameter bifurcation diagrams were not computed for fully timescale-separated parameter values.

The same parameter values were used for the $\mu$ bifurcation diagrams in Fig.~(\ref{fig:mu_lines}) except that $\beta_F = 25$ was set to ensure that the $\mu$ branch point occurs at a reasonable positive value. The $\mu$ bifurcation notably has two highly linear regimes which can be computed to high accuracy as demonstrated in Appendix~(\ref{sec:mu_lines}).

\section{Dynamics of the spatial stochastic FGBA model}\label{sec.spatial}
\subsection{Simulation algorithms and codes}
The spatial FGBA model was simulated in MATLAB using the Gillespie algorithm~\cite{gillespie1977exact}. All code used in this project can be accessed at \url{https://github.com/patterd2/vegetation_fire_models}. Testing and validation of the simulation codes is outlined in Appendix \ref{app.testing}.

\subsection{Grassland without forest}
Unless otherwise noted, all simulations in this section and the next section were run with $L = 1$ and $N = 2000$. Parameter values were identical to those used in the timescale-separated bifurcation diagram analysis, with the widths of all spreading kernels set to $\sigma_F=\sigma_B=\sigma_G=0.05$. The mean-field analysis predicts that the GBA model will exhibit a single ash-dominated stable fixed point. On longer timescales, we typically observe fluctuations about the mean-field dynamics in most parameter regimes, as expected, but the stochastic spatial model displays more complex dynamics on shorter timescales, and hence we focus on this case. Transient dynamics play a crucial role in ecological applications by capturing short- to medium-term ecosystem responses that differ from long-term equilibria, especially following disturbances or environmental change~\cite{hastings2001transient,hastings2018transient}. In the context of forest–savanna ecosystems, these dynamics may help to explain shifts in vegetation structure, species composition, and fire regimes observed in remotely sensed data on shorter timescales.

Using timescale-separated parameter values enables the GBA simulation to exhibit various behaviors that can only be observed appropriately on vastly differing timescales. For example, the dynamics of the cover proportions of the different states at a fire scale vs a grass time scale are shown in Fig.~(\ref{fig:fire_grass_timescale})AB. Fig.~(\ref{fig:fire_grass_timescale})A shows the profile of a single ignition event over a day, while Fig.~(\ref{fig:fire_grass_timescale})B shows multiple small fires and grass regrowth phases taking place over several months. Fig.~(\ref{fig:fire_grass_timescale})C shows a montage of the spatial progression of a typical fire ignition and spreading event occurring over a very short time (several hours). 

\begin{figure}[h] 
\centering  
\includegraphics[width=0.99\linewidth]{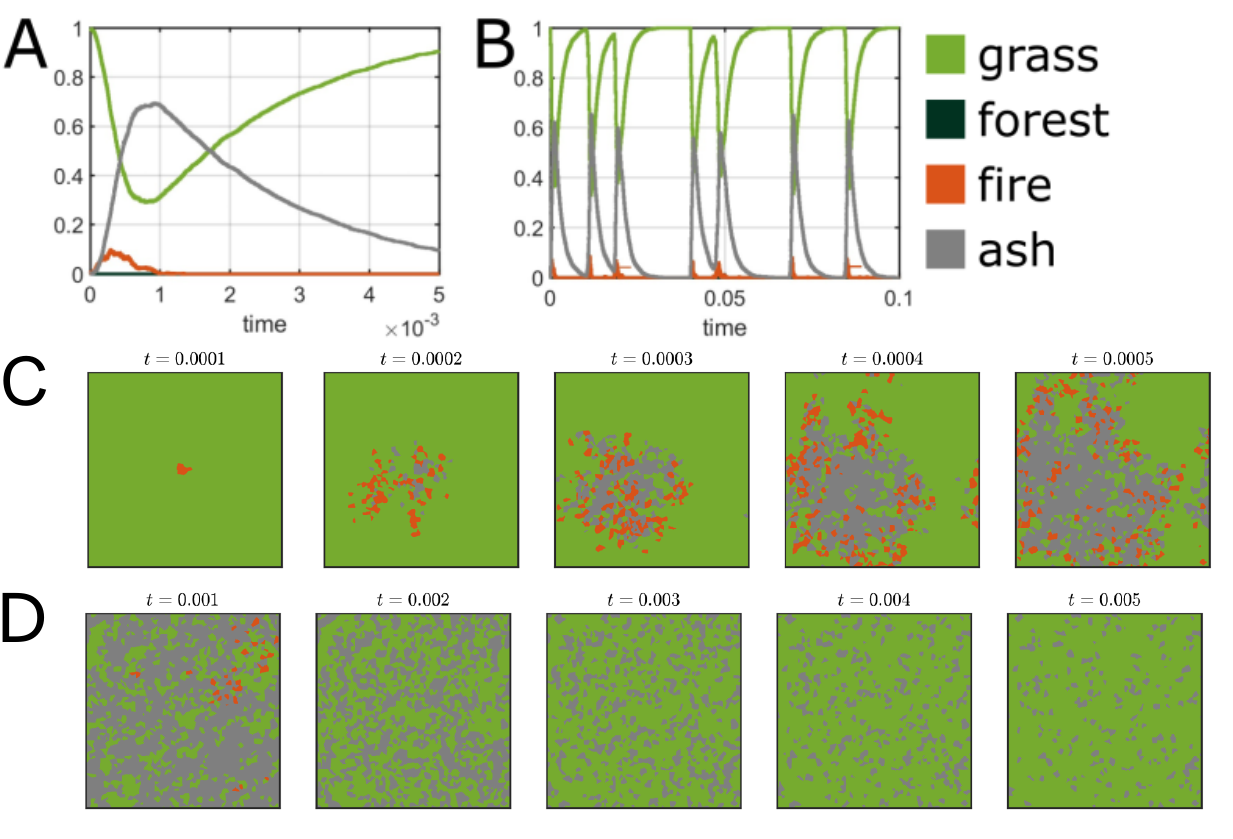}
\caption{A. System dynamics on the fire timescale with the proportions of space covered by each state on the y-axis. Note that a slight increase in burning area (fire ignition) is followed by a large increase in ash and a large decrease in grass. The grass cover then gradually recovers from the ash state. B. System dynamics on a grass timescale showing repeated big fire events followed by grass recovery. C. Montage of fire ignition and spreading in grassland on a fire timescale as shown in A. D. Montage of grass regrowth following a fire ignition event in grassland on a grass timescale as shown in A.}
\label{fig:fire_grass_timescale}
\end{figure} 

On the fire timescale, the spatial structure of a fire front spreading through space and leaving behind a region of ash is readily apparent. On the grass timescale, the fire dynamics are obscured, but the regrowth of grass following fire spreading events is easily observed. Compared to the fire spreading, grass regrowth occurs relatively homogeneously spatially as shown in Fig.~(\ref{fig:fire_grass_timescale})D.

\begin{figure}[h] 
\centering  
\includegraphics[width=0.99\linewidth]{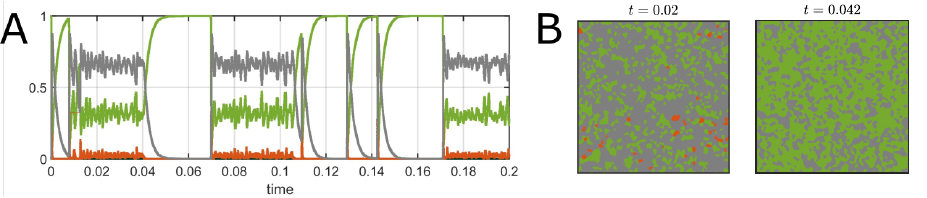}
\caption{A. Example of the system switching between large fires and mean-field-like behavior (i.e. fluctuations about an ash-dominated state). The y-axis is the proportion of space covered by each state. B. Snapshot of the spatial system in the stochastic analogue of the ash-dominated mean-field GBA steady state at time $t=0.02$ and a snapshot of the system in the grass-dominated quasi-steady state at time $t=0.042$. The color legend used is shown in Fig.~(\ref{fig:fire_grass_timescale}).}
\label{fig:mean_field_GBA}
\end{figure} 

In Fig.~(\ref{fig:fire_grass_timescale})B, several ignition events occur well separated in time on the grass timescale. This suggests that, in this parameter regime, the spatial GBA model possesses a largely grass quasi-steady state, distinct from the mean-field steady state, consisting of mostly ash. This is an especially interesting example of long-lived transient dynamics maintaining the system away from the mean-field prediction for extended periods. The spatial model can also demonstrate the mean-field grass steady state under certain conditions. For example, the spatial model can reach the mean-field GBA steady state for certain parameter values, which allows a uniform mixing of the states across the land area. For example, setting the fire quenching rate $q$ to $20,000$ instead of $15,000$ produces the behavior shown in Fig.~(\ref{fig:mean_field_GBA})A, with the system experiencing stochastic fluctuations that cause it to switch between regimes of big fire spreading events with spatial structure and mean-field behavior. The analogue of the mean-field GBA steady state is shown in Fig.~(\ref{fig:mean_field_GBA})B for $t=0.02$, while Fig.~(\ref{fig:mean_field_GBA})B with $t=0.042$ shows a grass-dominated quasi-steady state produced by transient and spatial dynamics not present in the mean-field model. A detailed investigation of the nature of these transient phenomena is beyond the scope of the present work, but it has been noted that high-dimensionality and spatial extent can promote long transients~\cite{hastings2018transient}, as observed here. In particular, the grass-dominated quasi-steady state in Fig.~(\ref{fig:mean_field_GBA})B takes time to build up a sufficiently dense and continuous layer of flammable cover to facilitate the large fire required to initiate the transition back to the ash-dominated state predicted by the mean-field. We hypothesize that this mechanism plays a key role in determining the duration of these extended transients.

\subsection{Grassland with Forest}
When studying grassland containing forest trees, we find additional new behavior at the forest timescale. In general, ignition events in grass cause large fires that propagate to an extent determined by the surrounding distribution of forest and destroy trees near the perimeter of the forest regions. During periods between fire events, the trees steadily regrow.

In general, there are two distinct outcomes over longer time periods. In the first scenario, the forest trees grow faster than the transitions that destroy them and eventually become so dense that fires can no longer propagate; thus, the high forest state is stably maintained. Alternatively, the grass fires are sufficiently destructive to reduce the forest proportion to a vanishing or near-vanishing proportion, where the forest can no longer spread noticeably. These two cases are illustrated in Fig.~(\ref {fig:forest_outcomes}) and confirm that our model produces the desired forest-savanna bistability in its complete spatial stochastic form. As forest cover decreases, fires become larger and more frequent. 

\begin{figure}[ht] 
\centering  
\includegraphics[width=1\linewidth]{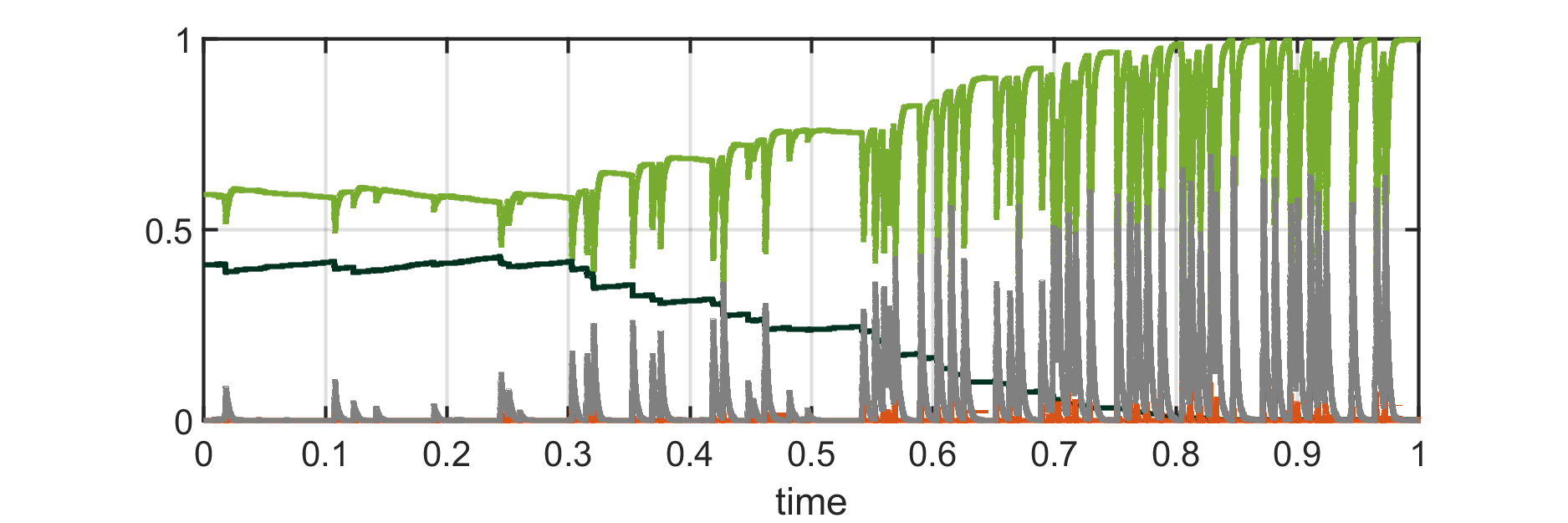}
\includegraphics[width=1\linewidth]{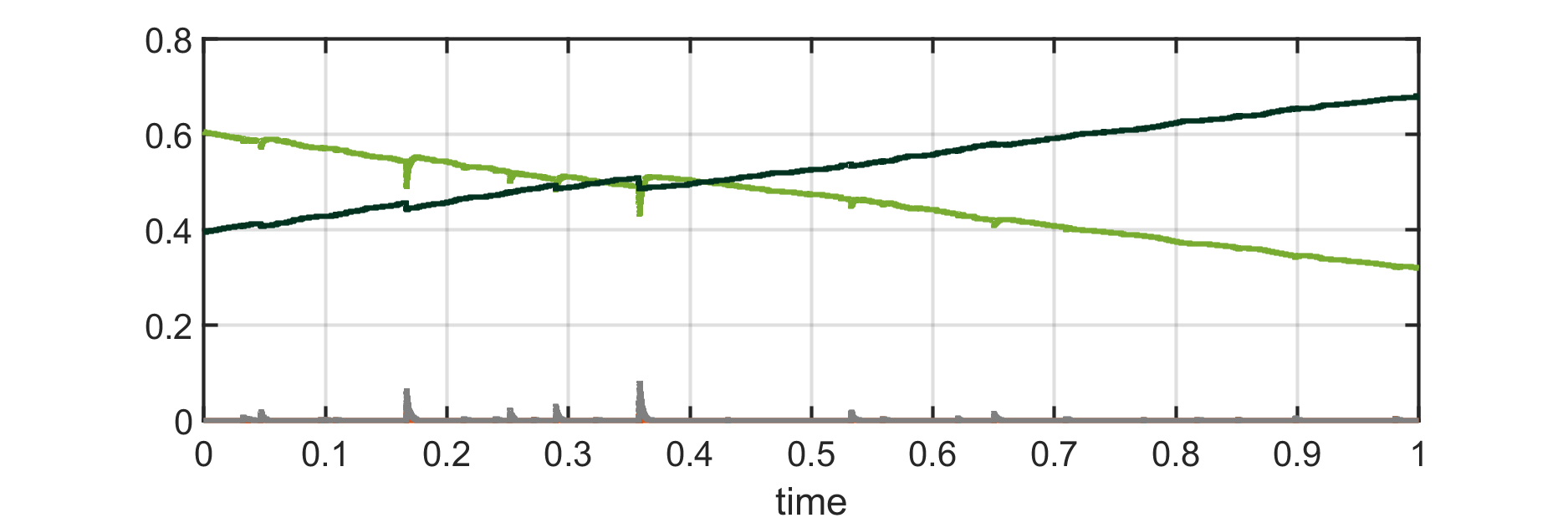}
\caption{Two different possible outcomes for forest (dark green line)  in the spatial-stochastic FGBA model; the proportion of space covered by each state is on the y-axis. Both simulations were run with the same parameter values and initial conditions. Forest was randomly distributed and occupied $40\%$ of sites. The remaining sites were all grass. The color legend is shown in Fig.~(\ref{fig:fire_grass_timescale}).}
\label{fig:forest_outcomes}
\end{figure} 

Due to the stochastic nature of the model and the underlying bistable structure, we cannot determine the long-term cover proportions at the outset; we may even observe noise-induced switching~\cite{weinan2021applied} between grass and forest-dominated steady states, which cannot be understood in terms of the mean-field approximation. In simulations starting at an initial condition of high forest cover, fire size also shows a noticeable increase as the forest cover decreases. A montage demonstrating the ability of forest cover to limit the extent of fire spreading is given in Fig.~(\ref{fig:forest_fire_montage}), corresponding to the fire occurring at the start of the simulation.

\begin{figure}[H] 
\centering  
\includegraphics[width=1\linewidth]{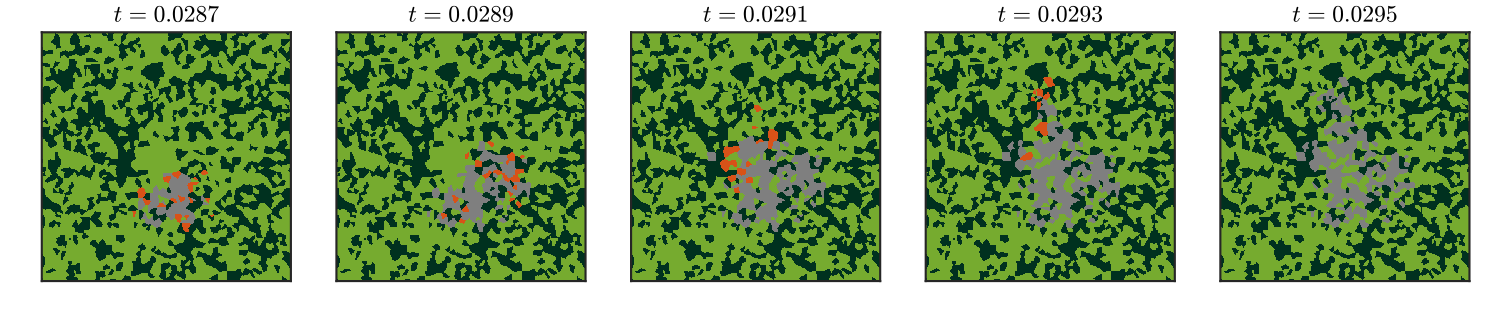}
\caption{A simulation demonstrating the ability of forest cover to limit the extent of fire spreading and specifically to stop smaller fires from growing into larger fires on the scale of the entire spatial domain (compare to Fig.~(\ref{fig:fire_grass_timescale})CD). The color legend is shown in Fig.~(\ref{fig:fire_grass_timescale}).}
\label{fig:forest_fire_montage}
\end{figure}

The top panel of Fig.~(\ref {fig:forest_outcomes}) highlights the dynamic nature of the open-canopy savanna-like state of the system, where frequent small to medium-scale fires are required to maintain the system state. Thus, our model in this setup could be used to distinguish between woody encroachment and natural fluctuations in the open canopy state. In contrast, we observe relatively small fluctuations in the forest-dominated state in the bottom panel of Fig.~(\ref {fig:forest_outcomes}) and hence our initial modeling suggests that empirically observed fluctuations in forest cover are likely to be of greater concern. It may be possible to quantify this effect with more detailed parameter fitting and comparison with empirical data. However, other factors such as rainfall seasonality, which will impact fire patterns, would need to be added to the modeling structure to see how they amplify or dampen the aforementioned fluctuations~\cite{accatino2013humid}.

\section{Conclusions}\label{sec.conclusion}

We have developed and analyzed a spatial stochastic model appropriate for studying fire-mediated forest-savanna bistability on short timescales where stochasticity and transient dynamics are expected to play crucial roles. Our work bridges the gap between highly resolved global vegetation models and minimalist models of forest-savanna bistability, such as the original Staver-Levin model. The mean-field (spatially implicit) version of our FGBA model provides a reasonably simple method to study the qualitative behavior of the system, but does not necessarily provide a full description of the dynamics of the spatial stochastic model. For instance, in the absence of forest trees, the stability of the grassland steady state under parameter regimes with sufficiently high forest mortality is observed in both the mean-field and spatial models. However, the mean-field model predicts only a single grassland steady state at unrealistically high ash coverage. In contrast, the stochastic spatial version of the model switches between a spatially homogeneous grassland steady state with high ash coverage and an additional quasi-steady state characterized by high grass coverage and occasional fire-spreading events well separated in time (see Fig.~(\ref{fig:fire_grass_timescale})). We hypothesize that these new transient dynamics are a product of the high-dimensional nature of the system and the ability of the spatial extent to block large fires until the grass layer has become sufficiently dense; these phenomena are beyond the scope of the mean-field model and hence require a separate, dedicated analysis to quantify fully. In a related example, the stability of the forest as a function of the parameter values showed qualitative agreement in both the mean-field and spatial models with increases in parameters such as $\varphi$ (the rate of fire spread over grass or ash), allowing the system to have a finite probability of reaching a high forest state. While the extant modeling literature has often relied on mean-field approximations, our work highlights the need to consider appropriately detailed models when estimating ecosystem resilience and stability from empirical data. In addition, simulations of our spatial FGBA model explicitly demonstrate the hypothesized mechanisms underpinning bistability in forest tree cover. In particular, occasional ignitions followed by rapid fire spreading can maintain low forest cover even at high forest spreading rates. Meanwhile, the inability of fire to spread in regions of dense tree cover maintains high forest cover even at high fire ignition and spread rates. 

Our proposed FGBA model is highly versatile and can be used to study a wide range of possible forest and grassland setups. For example, differences in soil quality could be modeled by choosing vegetation sites within $\Omega$ according to a non-uniform probability distribution. One could also model forest spread via heavy-tailed, non-Gaussian spreading kernels, which may be a more accurate model of forest spread than our current assumption that forest trees only spread locally~\cite{nathan2012dispersal}. One could also investigate the impact of non-spatially uniform forest distributions, for example, if the forest were distributed into distinct regions of high tree density separated by areas of low tree density or into shapes with varying perimeter-area ratios. An essential addition to consider in future work and developments of this framework will be rainfall seasonality and its associated consequences for fire ignition. Forest-savanna ecosystems are subject to yearly wet and dry seasons that significantly impact tree growth and the relative humidity of the ecosystem, hence making the system more or less fire-prone at certain times of the year~\cite{accatino2013humid}. This will likely alter the dynamics presented in section \ref{sec.spatial}, but we believe the present analysis is a crucial baseline against which to benchmark results incorporating seasonality to isolate the new effects introduced by seasonal forcing.

Our model aims to alleviate several limitations of similar mathematical models of forest-savanna ecosystems, while retaining a reasonable degree of mathematical tractability. In particular, our model is spatially explicit, models dynamics on all time scales (fire, grass, and forest), allows an arbitrary distribution of vegetation sites, and incorporates arbitrary spreading kernels, enabling the study of non-isotropic or long-distance spreading. However, to avoid excessive mathematical and computational complexity, our model has some limitations. For instance, our model considers only fire and vegetation spread as transition processes in this paper, although our overarching framework can incorporate other processes. Other models consider additional factors, such as herbivory, soil layer quality, and competition between species~\cite{abades2014fire,holdo2022sapling}, which could be added to our mathematical framework at the cost of additional complexity. Moreover, we could make our model more realistic by expanding the state-space to include different savanna tree types at different life stages, as in other related models~\cite{touboul2018complex}. The key role of transient dynamics and stochasticity will likely remain in such extensions. Still, it would undoubtedly be important to extend the present work to these more realistic settings to relate more closely to empirical observations and hence further our understanding of these precious and endangered ecosystems.

\section*{Statements and Declarations}

\subsection*{Acknowledgements} 

Denis Patterson and Simon Levin are grateful for the support of NSF DMS-1951358.

\subsection*{Competing Interests} The authors have no competing interests to declare.

\subsection*{Data Availability} The codes used in this paper are available at \url{https://github.com/patterd2/vegetation_fire_models}.

\noindent

\begin{appendices}

\section{Mathematical Proofs}\label{secA1}
\begin{proof}[Proof of Proposition (forward invariance of the GBA subspace)]
Equivalently $(G(t), B(T))$ remains contained in the closed triangle $\mathcal{T}_{GBA} = \{ (G,B) \in \mathbb{R}^2: G \geq 0, B \geq 0, G + B \leq 1\}$ at all times. This is easily shown by observing that the vector field $(\dot{G}, \dot{B})$ points towards the interior of $\mathcal{T}_1$ everywhere on $\partial \mathcal{T}_1$:
\begin{align*}
    \begin{cases}
        G = 0 & \Rightarrow  \dot{G} = \gamma(1- B) \geq 0\\
        B = 0 & \Rightarrow \dot{B} = \Phi_G(G)G \geq 0\\
        G + B = 1 & \Rightarrow \dot{G} + \dot{B} = -qB \leq 0.
    \end{cases}
\end{align*}
\end{proof}
\begin{proof}[Proof of Proposition \ref{prop_2} (existence and uniqueness of the GBA steady state)]
We define for convenience
\begin{align*}
    \mathcal{F}(G) \equiv (1-G)\frac{\gamma q}{\gamma + q} \bigg(\frac{1}{G}-\frac{\beta_G}{q}\bigg).
\end{align*}
Then $\overline{G}$ is a steady state if and only if $\mathcal{F}(\overline{G}) = \Phi_G(\overline{G})$. Let us first consider the case where $q < \beta_G$. Then in the interval $[0,1]$, $\mathcal{F}$ has exactly two roots located at $\frac{q}{\beta_G}$ and $1$. Furthermore, $\mathcal{F} > 0$ on $(0, \sqrt{q/\beta_G})$ and $\mathcal{F} < 0$ on $(\sqrt{q/\beta_G},1)$. Since $\Phi_G > 0$ on $(0,1]$ and $0 \leq \Phi_G(0) \leq \infty$ then any roots of $\mathcal{F} - \Phi_G$ in $[0,1]$ must occur in the interval $(0, q/\beta_G)$. Next note that the derivative of $\mathcal{F}$ is
\begin{align*}
    \mathcal{F}'(G) = \frac{\gamma q}{\gamma + q}\bigg(\frac{\beta_G}{q} - \frac{1}{G^2} \bigg)
\end{align*}
and has a single root in $[0,1]$ located at $\sqrt{q/\beta_G}$. Furthermore, $\mathcal{F}' < 0$ on $(0, \sqrt{q/\beta_G})$. Since $(0, q/\beta_G) \subset (0, \sqrt{q/\beta_G})$ it follows that $\mathcal{F}$ is strictly monotonically decreasing on $(0, q/\beta_G)$. Next since $\Phi_G$ is monotonically increasing on $(0, q/\beta_G)$ then $\mathcal{F} - \Phi_G$ is strictly monotonically decreasing on $(0, q/\beta_G)$. Since $\mathcal{F} - \Phi_G$ is continuous on $\mathbb{R}^+$ and $\lim_{G \rightarrow \infty}(\mathcal{F} - \Phi_G)(G) = +\infty$ while $(\mathcal{F}- \Phi_G)(q/\beta_G) = -\Phi_G(q/\beta_G) < 0$ then $\mathcal{F} - \Phi_G$ must have exactly one root in $(0, q/\beta_G)$. It follows that there is a unique steady state $\overline{G} \in [0,1]$. 

In the case where $q \geq \beta_G$ then $\mathcal{F}$ has exactly one root at $1$ in the interval $[0,1]$ and $\mathcal{F} > 0$ on $(0, 1)$. By the same argument as before $\mathcal{F}' < 0$ on $(0,\sqrt{q/\beta_G})$ so $\mathcal{F}$ is strictly monotonically decreasing on $(0,1)$. The rest of the argument is the same as the $q < \beta_G$ case with $q/\beta_G$ replaced by $1$. Thus, a unique GBA steady state always exists for any choice of parameters for the system. See Fig.~(\ref{fig:GBA}) for a graphical illustration. 

\begin{figure}[ht] 
\centering  
\includegraphics[width=0.95\linewidth]{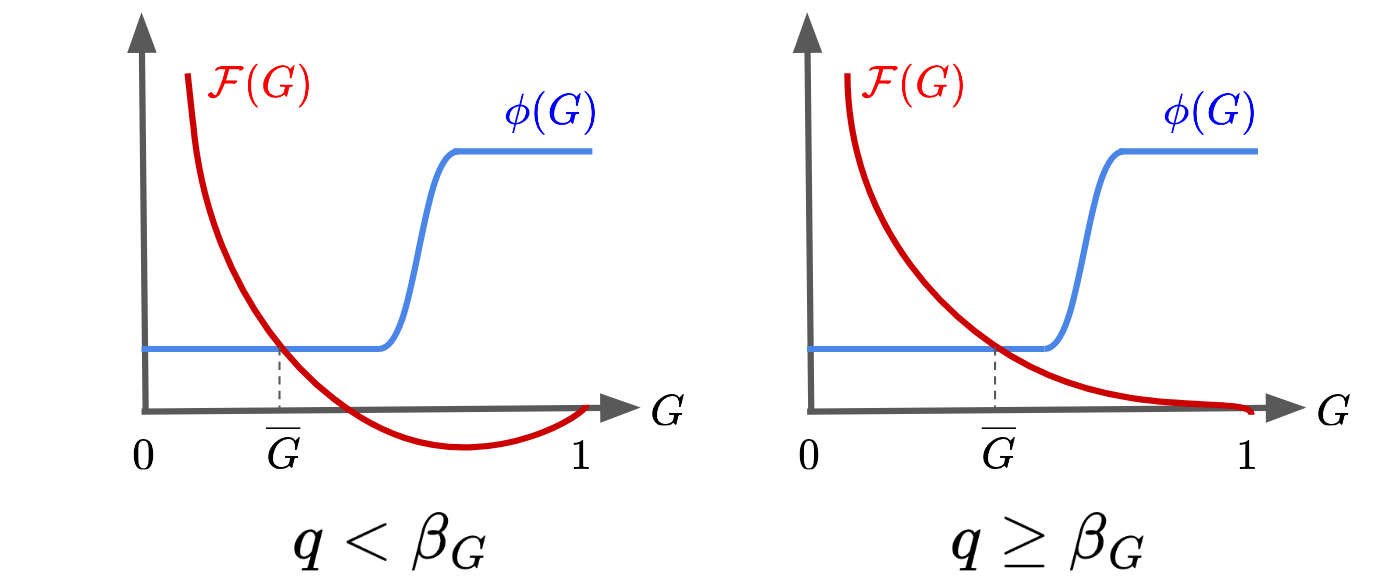}
\caption{Plots of $\mathcal{F}$ and $\Phi_G$ in the $q < \beta_G$ case (left) and the $q \geq \beta_G$ case (right)
\label{fig:GBA}} 
\end{figure}

Graphically, we note that $\overline{G}$ increases as $\frac{\gamma q}{\gamma + q}$ increases. Noting that $\partial_{\gamma}(\frac{\gamma q}{\gamma + q}) > 0$ and $\partial_{q}(\frac{\gamma q}{\gamma + q}) > 0$ when $\gamma, q > 0$ it follows that increasing either $\gamma$ or $q$ will increase the proportion of grass in the steady state. This is logical since increasing $\gamma$ allows faster regrowth of grass from ash, and increasing $q$ allows faster fire quenching into ash from which grass can regrow. Similarly, increasing $\frac{q}{\beta_G}$ also increases $\overline{G}$. This also follows intuition: the case of increasing $q$ has already been discussed, while decreasing $\beta_G$ reduces the rate at which grass burns due to fire spread. We also note that increasing $g_0$ or $g_1$ reduces $\overline{G}$ as expected since doing so increases the flammability of grass. Lastly, we note that the presence of the GBA steady state depends on the non-decreasing non-tonicity of $\Phi_G$. Thu,s if $\Phi_G$ were constant, the result of the system containing a single unique GBA steady state would still hold. 
\end{proof}

\begin{proof}[Proof of Proposition \ref{prop.fwd.FGBA} (forward invariance for the nonspatial FGBA model)]
    Analogously to the GBA system, we check for flow invariance of the tetrahedron $\mathcal{T}_{FGBA} = \{(F,G,B) \, | \, F \geq 0, G \geq 0, B \geq 0, F +G+B \leq 1\}$ to ensure that any trajectory starting at an ecologically relevant condition remains ecologically relevant at all times. 
    \begin{align*}
    \begin{cases}
        F = 0 & \Rightarrow \dot{F} =  0\\
        G = 0 & \Rightarrow \dot{G} = \gamma(1-F-B)+\mu F  \geq 0 \\
        B = 0 & \Rightarrow \dot{B} = \Phi_G(G)G + \Phi_F(G)F \geq 0  \\
        F +G+B = 1 & \Rightarrow \dot{F} +\dot{G} + \dot{B} = -qB \leq 0,  
    \end{cases}
\end{align*}
so $\mathcal{T}_{FGBA}$ is invariant, as desired.
\end{proof}

\section{Computation of equilibrium burning cover value}\label{secA2}
We now examine the function $\overline{B}(F)$, which generates an equilibrium burning cover value for any input forest value (including $F=0$). To find $\overline{B}(F)$, set the four equations in system~(\ref{eqn:mf}) to $0$, forming an essentially three-dimensional system. From the equation for $\dot{A}$ we obtain
\begin{align*}
    \overline{G}(F,B) &= 1 - F - \Big( \frac{\gamma + \varphi F + q}{\gamma + \varphi F}\Big) B.
\end{align*}
Plugging this into the equation for $\dot{B}$ then gives
\begin{align*}
    0 &= c_0(F) + c_1(F) B + c_2(F) B^2
\end{align*}
where 
\begin{align*}
    c_0(F) &= g_0(1-F) + f_0F \\
    c_1(F) &= \beta_G(1-F) - g_0 \Big( \frac{\gamma + \varphi F + q}{\gamma + \varphi F} \Big) + \beta_F F - q \\
    c_2(F) &= -\beta_G \Big( \frac{\gamma + \varphi F + q}{\gamma + \varphi F} \Big)
\end{align*}
where we work in the assumption that $F \approx 0$. A simple application of the quadratic formula then gives
\begin{align*}
    \overline{B}(F) &= \frac{-c_1(F) \pm \sqrt{c_1^2(F) - 4c_0(F)c_2(F)}}{2c_2(F)}.
\end{align*}
Next note that as $g_0 \rightarrow 0$ then
\begin{align*}
    \overline{B}(0) &= \frac{(\beta_G - q) \mp (\beta_G -q)}{2 \beta_G(\frac{\gamma + q}{\gamma})}
\end{align*}
where we assume that $\beta_G > q$. Comparing this expression to the GBA steady state with $g_0 \rightarrow 0$ in Eq.~(\ref{eqn:GBA}) then indicates that we must choose the $+$ sign. Next, using Mathematica we can expand $\overline{B}(F)$ in powers of $F$: 
\begin{align*}
    \overline{B}(F) &= d_0 + d_1 F + d_2 F^3 + \mathcal{O}(F^3)
\end{align*}
where
\begin{align*}
    d_0 &= \frac{(\beta_G - q) \gamma}{\beta_G(q+\gamma)} + \Big(\frac{q}{\beta_G(\beta_G-q)} \Big)g_0 + \mathcal{O}(g_0^2) \\
    d_1 &= \frac{(\beta_f - \beta_G)\gamma(q + \gamma) + q(\beta_G-q)\varphi}{\beta_G(q+\gamma)^2} + \Big(\frac{1}{\beta_g-q} \Big)f_0 + \Big(\frac{(q-\beta_F)(\beta_F - \beta_G)}{(q-\beta_G)^3} \Big)g_0 \\
    & \quad \quad \quad \quad \quad + \mathcal{O}(g_0^2, g_0 f_0, f_0^2) \\
    d_2 &= \frac{q\varphi((\beta_F-\beta_G)(q+\varphi) + (q-\beta_G)\varphi)}{\beta_G(q+\gamma)^3} + \mathcal{O}(g_0, f_0).
\end{align*}
Using our orders of magnitude approximation $\beta_G \approx q \gg \gamma \gg \varphi \approx 1$ and $\beta_G \gg \beta_F$ we find that 
\begin{align*}
    d_0 &\approx \frac{\gamma}{q} + \frac{1}{q} g_0, \quad 
    d_1 \approx \frac{\gamma}{q} + \frac{1}{q} (f_0 + g_0), \quad 
    d_2 \approx \frac{\varphi}{q},
\end{align*}
so $d_0 \approx d_1 \gg d_2$, and we can reasonably approximate
\begin{align*}
    \overline{B}(F) \approx d_0 + d_1 F,
\end{align*}
and it is also apparent that
\begin{align*}
    \overline{B}(0) = d_0, \quad 
    \overline{B}'(0) = d_1.
\end{align*}
We can perform a similar analysis for $F \approx 1$. In this case, we write
\begin{align*}
    \overline{B}(F) = \tilde{d}_0 + \tilde{d}_1(F - 1) + \mathcal{O}((F-1)^2)
\end{align*}
where
\begin{align*}
    \tilde{d}_0 &= \frac{f_0}{q-\beta_F}+ \mathcal{O}(g_1^2, g_1f_1, f_1^2), \\ 
    \tilde{d}_1 &= \Big( \frac{q-\beta_G}{(q-\beta_F)^2} \Big) f_1 + \Big(\frac{1}{\beta_f - q}\Big) g_1 + \mathcal{O}(g_1^2, g_1f_1, f_1^2).
\end{align*}

\section{Simulation Testing} \label{app.testing}
Below, we perform some sense checks on the model to ensure that the simulation codes are bug-free and behaves as expected. We test each transition rate separately by introducing artificial situations with readily predictable behaviour and comparing the simulation results with analytic predictions. 

\subsection{Testing non-spatial transitions}
The simplest transitions to test are the spontaneous, non-spatial transitions governed by the parameters $q$ ($B \rightarrow A$), $\gamma$ ($A \rightarrow G)$ and $\mu$ ($F \rightarrow G$) corresponding to fire quenching, grass regrowth, and tree mortality, respectively. All testing was performed with $L = 100$ and $N = 500$. 

To test the fire quench rate we began the simulation with all states in the burning state and removed grass regrowth by setting $\gamma = 0$. The times at which each burning state transitioned into the ash state were recorded at various values of $q$. The quench times are expected to follow a Poisson distribution with rate $q$, giving an expected quench time of $\frac{1}{q}$ with variance $\frac{1}{q}$. 

Similarly, to check the grass regrowth rate from ash we began with all sites in the ash state. We set $g_0 = g_1 = 0$ to remove spontaneous grass fires. The times at wich each ash site transitoned into a grass state was recorded for various values of $\gamma$. As expected, he regrowth times followed a Poisson distribution with rate $\gamma$ and expected regrowth time $\frac{1}{\gamma}$ with variance $\frac{1}{\gamma}$. 

\begin{figure}[ht] 
\centering  
\includegraphics[width=0.32\linewidth]{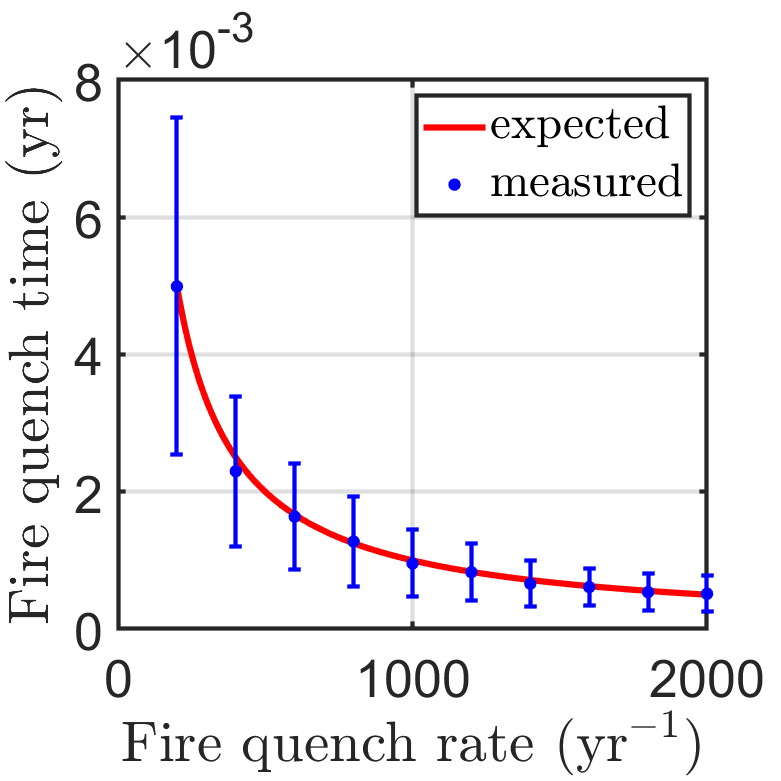}
\includegraphics[width=0.32\linewidth]{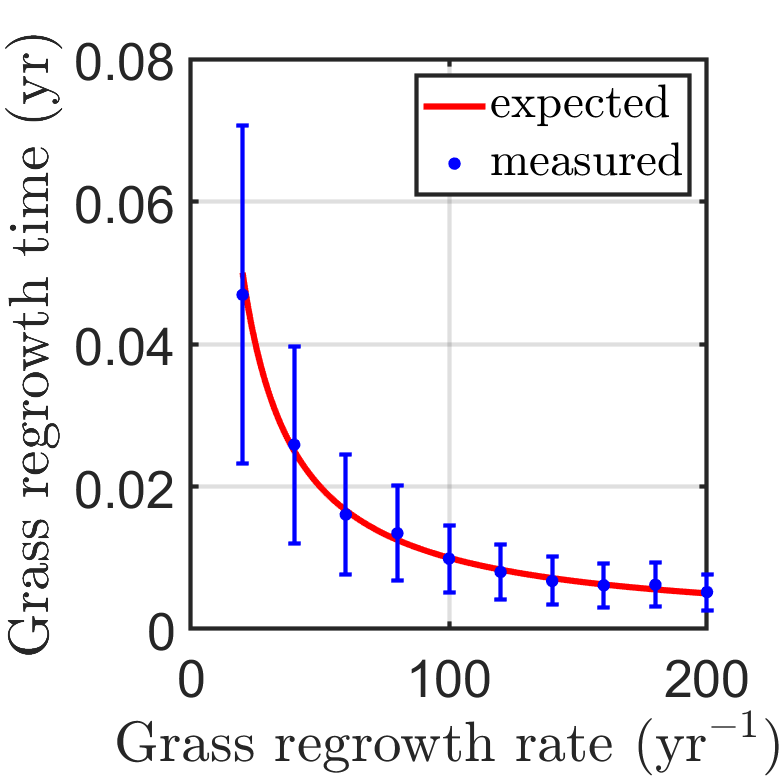}
\includegraphics[width=0.32\linewidth]{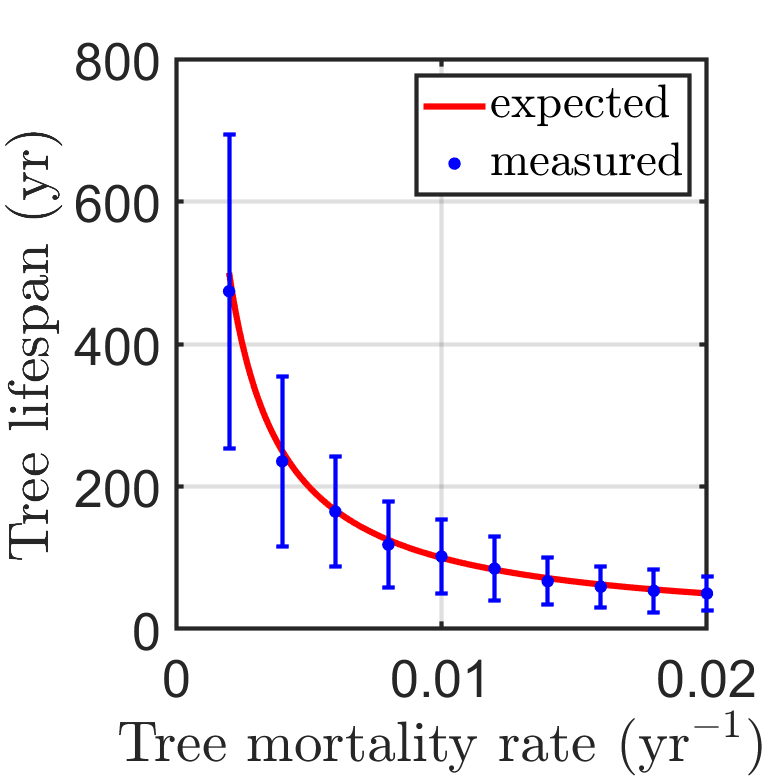}
\caption{Testing the fire quench (left), grass regrowth (center), and tree mortality (right) transitions}
\label{fig:fire_quench}
\end{figure} 

Lastly, to check forest mortality transitions we began the simulation with all sites in the forest state and removed forest spread by setting $\varphi_G = \varphi_A = 0$. Grass and forest fires were also eliminated by setting $g_0 = g_1 = f_0 = f_1 = 0$. The times at which each forest site transitioned into a grass state was recorded at various values of $\mu$. The tree lifetimes are expected to follow a Poisson distribution wtih rate $\mu$ giving an expected lifetime of $\frac{1}{\mu}$ with variance $\frac{1}{\mu}$.

\subsection{Testing spatial transitions}
We next tested the four directly spatial transitions in the FGBA model: the spread of forest through grass and through ash and the spread of fire through grass and through forest. 

Fire spread through grass was tested introducing the initial condition where a single randomly chosen vegetation site in the domain would be in a grass state while the remaining sites are all randomly assigned to be a either burning or ash states according to some fixed probability $p_{\text{ash}}$ of the state being ash. Further, $g_0 = g_1 = \gamma = q = 0$ was set to remove spontaneous grass ignitions and grass regrowth and prevent fire quenching. The burn rate at the grass site is then expected to be approximately $\beta_G p_{\text{ash}}$.  The time at which the grass site transitioned to a burning state was then measured at various values of $\beta_G$ and $p_{\text{ash}}$. The average of $600$ trials for each pair of values $(\beta_G, p_{\text{ash}})$ is plotted in Fig.~(\ref{fig:test_fire_spread}). 

Fire spread through forest was tested in an analogous situation where instead of a single grass site, a single forest site was used and $f_0 = f_1 = \mu = 0$ was set to eliminate all possible transitions except the forest to burning transition. The burn rate at the forest site is then expected to be approximately $\beta_F (1-p_{\text{ash}})$. The time at which the grass site transitioned to a burning state was then measured at various values of $\beta_F$ and $p_{\text{ash}}$. The average of $600$ trials for each pair of values $(\beta_F, p_{\text{ash}})$ is plotted in Fig.~(\ref{fig:test_fire_spread}).

\begin{figure}[ht] 
\centering  
\includegraphics[width=0.49\linewidth]{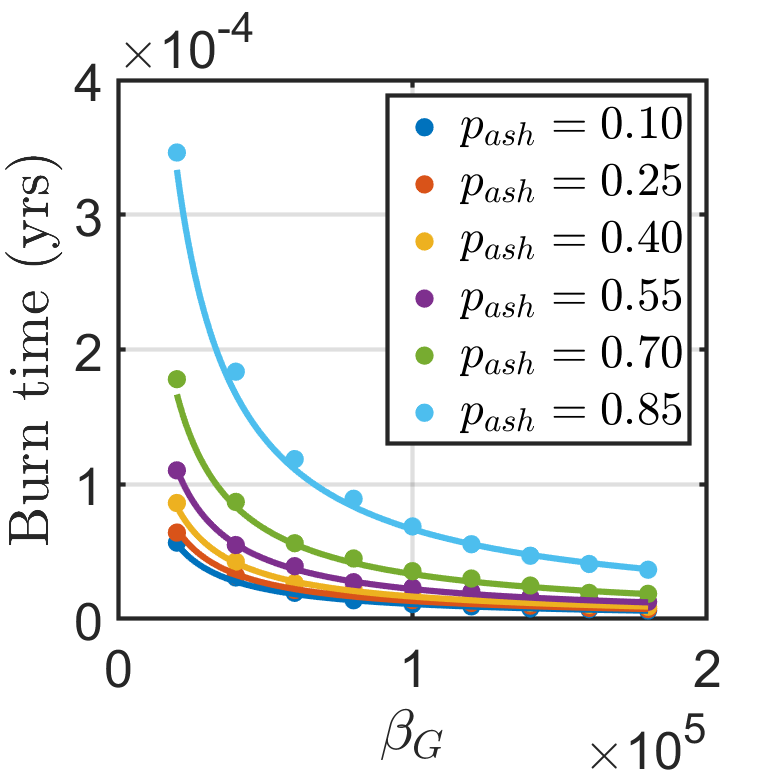}
\includegraphics[width=0.49\linewidth]{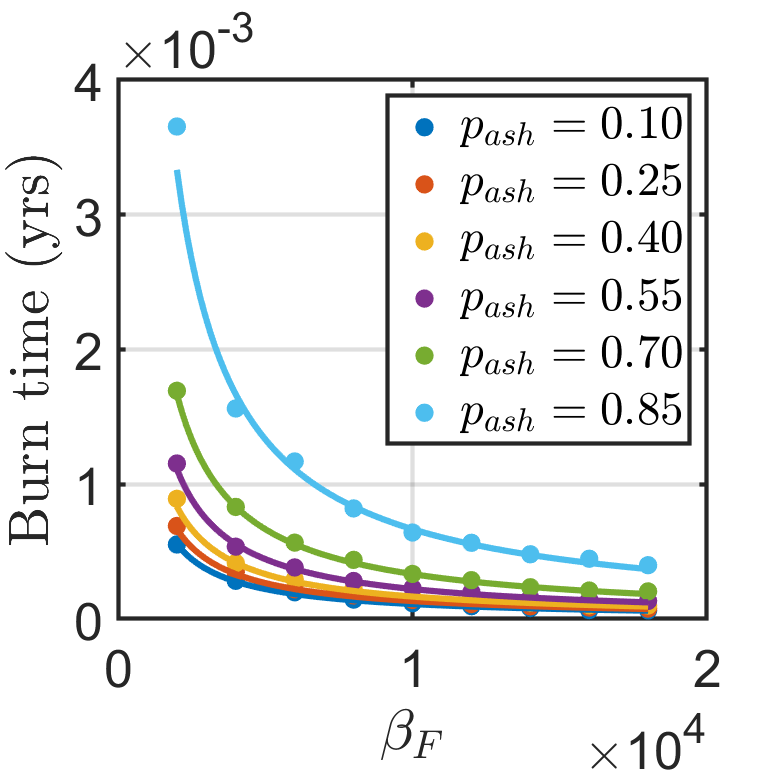}
\caption{Testing the spatial transition of grass burning (left) and forest burning (right) with N = 2000.}
\label{fig:test_fire_spread}
\end{figure} 

Next, the spread of forest through grass was tested by randomly choosing a single vegetation site to be grass; the remaining sites are randomly chosen to be in the forest state with probability $p_{\text{forest}}$ and are otherwise in the ash state. To prevent all other possible transitions besides the grass to forest transition we set $g_0 = g_1 = f_0 = f_1 = q = \gamma = \mu = \varphi_A = 0$. Then the rate at which the grass site should transition to forest is approximately $\varphi_G p_{\text{forest}}$. The time at which the grass site transitioned to a forest state was then measured at various values of $\varphi_G$ and $p_{\text{forest}}$. The average of five trials for each pair of values $(\varphi_G, p_{\text{forest}})$ is plotted in Fig.~(\ref{fig:test_forest_spread}).

Lastly, the spread of forest through ash was tested by randomly choosing a single vegetation site to be ash while the remaining sites are randomly chosen to be in the forest state with probability $p_{\text{forest}}$ and are otherwise in the grass state. To prevent all other possible transitions besides the ash to forest transition, we set $g_0 = g_1 = f_0 = f_1 = q = \gamma = \mu = \varphi_G = 0$. Then the rate at which the ash site should transition to forest is approximately $\varphi_A p_{\text{forest}}$. The time at which the grass site transitioned to a forest state was then measured at various values of $\varphi_A$ and $p_{\text{forest}}$. The average of five trials for each pair of values $(\varphi_A, p_{\text{forest}})$ is plotted in Fig.~(\ref{fig:test_forest_spread}).

\begin{figure}[ht] 
\centering  
\includegraphics[width=0.49\linewidth]{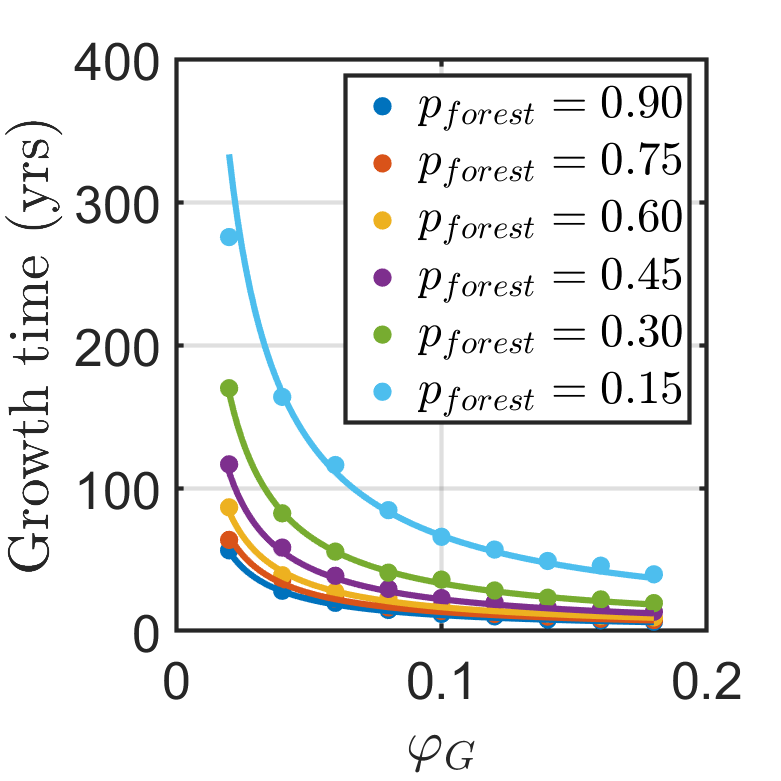}
\includegraphics[width=0.49\linewidth]{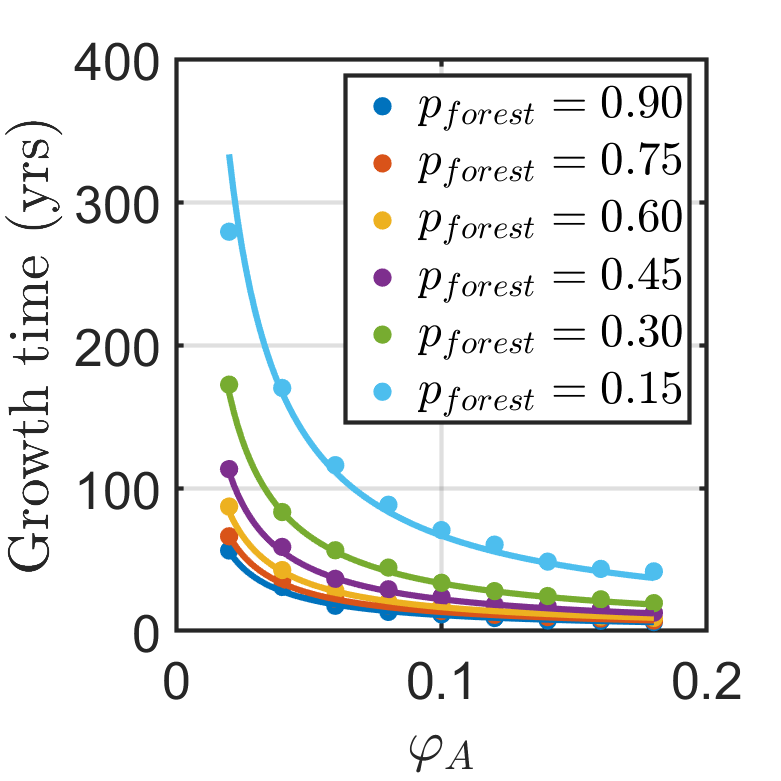}
\caption{Testing of the spatial transition of forest spread into grass (left) and ash (right).}
\label{fig:test_forest_spread}
\end{figure}

\section{Computation of linear regimes in non-fire forest mortality bifurcation plot} \label{sec:mu_lines}
Note the largely linear regimes within the non-fire forest mortality ($\mu$) bifurcation plots aside from an elbow at $F \approx 0.6$. To determine exact expressions for these two regimes, we note that the equation for $\dot{F}$ in system~(\ref{eqn:mf}) can only vanish for $F>0$ if
\begin{align*}
    \varphi(1-F) - (\varphi + \beta_F)(\overline{B}(F))- \Phi_F(F) - \mu = 0
\end{align*}
For $F < 0.6$ we approximate $\overline{B}(F)$ as a linear function of $F$ (see App.~\ref{secA2}) so that we obtain a line with x-intercept
\begin{align*}
    \mu = \varphi - (\varphi + \beta_F)\Big( \frac{(\beta_G-q) \gamma}{\beta_G(q+\gamma)}\Big) - f_0
\end{align*}
and slope
\begin{align*}
    m = \frac{\beta_G(q + \gamma)^2}{\beta_F(\beta_F-\beta_G)\gamma(q + \gamma) + q \beta_F(\beta_G-q)\varphi}.
\end{align*}
Similarly, for $F > 0.6$ an analogous calculation gives a line with a $y$-intercept at
\begin{align*}
    F &= \frac{\varphi - (\varphi + \beta_F)(\tilde{d}_0 - \tilde{d}_1) - f_1}{\varphi + (\varphi + \beta_F)\tilde{d}_1}
\end{align*}
and slope
\begin{align*}
    m = \frac{1}{-\varphi - (\varphi + \beta_F)\tilde{d}_1}.
\end{align*}
To understand the origin of the elbow at $F \approx 0.6$ we note that $\overline{B}'(F)$ is undefined when the discriminant $c_1(F)^2 - 4c_0(F)c_2(F)$ vanishes. To zeroth order in $f_0$ and $g_0$ we then find that $\overline{B}'(F)$ is undefined when 
\begin{align*}
    F^* &= \frac{q - \beta_G}{\beta_G - \beta_F}.
\end{align*}
Furthermore, to zeroth order in $f_0$ and $g_0$
\begin{align*}
    \overline{B}'(F) = \begin{cases}
        0 & F > F^* \\
        \frac{-q^2\varphi + q \beta_G(-\gamma + \varphi - 2F \varphi)+(\beta_F-\beta_G)(\gamma + F \varphi)^2 + q\beta_F(\gamma + 2 \varphi)}{\beta_G(q + \gamma + F \varphi)^2} & F < F^*.
    \end{cases}
\end{align*}
Note that all $F$-dependent terms are suppressed since $q \gg \gamma \gg \varphi$ so $\overline{B}'(F)$ is approximately constant in its two regimes, so $\overline{B}(F)$ is approximately linear. 




\end{appendices}


\bibliography{sn-bibliography}
\bibliographystyle{abbrv}

\end{document}